\documentclass[a4paper, 11pt]{article}
\usepackage{amssymb,amsmath,amsthm,lmodern,cancel}
\usepackage{amsfonts}
\usepackage{mathrsfs}  
\usepackage{subcaption}

\usepackage[dvipsnames]{xcolor}
\usepackage{graphicx}
\usepackage[english]{babel}
\usepackage[T1]{fontenc}
\usepackage{textcomp}
\usepackage[shortlabels]{enumitem}
\usepackage{thmtools, mathtools,verbatim} \usepackage{upgreek}
\usepackage{thm-restate}
\usepackage{titlesec} 

\usepackage[colorlinks=false,hyperindex=true]{hyperref}
\usepackage[noabbrev,capitalize]{cleveref}

\usepackage[font=normal]{caption}

\usepackage{thmtools}
\usepackage{amssymb}

\usepackage[bottom]{footmisc}

\usepackage{tikz}
\usetikzlibrary{math}
\usetikzlibrary{calc,fit,shapes}

\tikzset{vertex/.style={circle,draw=black,fill=gray,inner sep=3pt}
,edgeVertex/.style={circle,draw=black,fill=white,inner sep=3pt}
,edgeVertexSpinPlus/.style={circle,draw=black,fill=red!10,inner sep=3pt,append after command={(\tikzlastnode.west) edge (\tikzlastnode.east) (\tikzlastnode.south) edge (\tikzlastnode.north)}}
,edgeVertexSpinMinus/.style={circle,draw=black,fill=blue!10,inner sep=3pt,append after command={(\tikzlastnode.west) edge (\tikzlastnode.east) }}
,face/.style={diamond,draw=black,fill=white,inner sep=2.5pt,append after command= {(\tikzlastnode.west) edge (\tikzlastnode.east) (\tikzlastnode.south) edge (\tikzlastnode.north)}}
,clusterFace/.style={diamond,draw=black,fill=white,inner sep=2.5pt,append after command= {(\tikzlastnode.west) edge (\tikzlastnode.east) (\tikzlastnode.south) edge (\tikzlastnode.north)},append after command= {(\tikzlastnode.west) node[circle,fill=black,inner sep=1pt]{}} {(\tikzlastnode.east) node[circle,fill=black,inner sep=1pt]{}} {(\tikzlastnode.south) node[circle,fill=black,inner sep=1pt]{}} {(\tikzlastnode.north) node[circle,fill=black,inner sep=1pt]{}}}
,clusterFaceDobOut/.style={diamond,draw=black,fill=white,inner sep=2.5pt,append after command= {(\tikzlastnode.west) edge (\tikzlastnode.east) (\tikzlastnode.south) edge (\tikzlastnode.north)},append after command= {(\tikzlastnode.west) node[circle,fill=black,inner sep=1pt]{}} {(\tikzlastnode.east) node[circle,fill=green,inner sep=1pt]{}} {(\tikzlastnode.south) node[circle,fill=black,inner sep=1pt]{}} {(\tikzlastnode.north) node[circle,fill=black,inner sep=1pt]{}}}
,clusterFaceDobIn/.style={diamond,draw=black,fill=white,inner sep=2.5pt,append after command= {(\tikzlastnode.west) edge (\tikzlastnode.east) (\tikzlastnode.south) edge (\tikzlastnode.north)},append after command= {(\tikzlastnode.west) node[circle,fill=black,inner sep=1pt]{}} {(\tikzlastnode.east) node[circle,fill=black,inner sep=1pt]{}} {(\tikzlastnode.south) node[circle,fill=green,inner sep=1pt]{}} {(\tikzlastnode.north) node[circle,fill=black,inner sep=1pt]{}}}
,dualEdge/.style={dashed}
,clusterLongEdge/.style={black}
,dobEdge/.style={dashed,very thick, green}
,dobFace/.style={face,fill=green}
,contour/.style={very thick}
,peierls/.style={very thick,blue,inner sep=0pt}
,excludedEdge/.style={contour,red}
}

\DeclareRobustCommand{\tikzvertex}[1]{\tikz[baseline=default]{ \node [#1] {};}}
\usetikzlibrary{external}

\def\examplen{9}

\DeclareMathOperator{\diam}{diam}
\DeclareMathOperator{\supp}{supp}
\DeclareMathOperator{\dist}{dist}

\newcommand{\SLE}[1]{\ensuremath{\textup{SLE}({#1})}}

\newcommand{\ul}[1]{{\ensuremath{\underline{#1}}}}
\newcommand{\lis}[1]{{\ensuremath{\overline{#1}}}}

\newcommand{\bulk}{\ensuremath{\textup{bulk}} }

\newcommand{\state}[1]{\left\langle{#1} \right\rangle}

\newcommand{\bC}{\ensuremath{\mathbb{C}}}
\newcommand{\bE}{\ensuremath{\mathbb{E}}}

\newcommand{\bP}{\ensuremath{\mathbb{P}}}

\newcommand{\bR}{\ensuremath{\mathbb{R}}}

\newcommand{\bZ}{\ensuremath{\mathbb{Z}}}

\newcommand{\bfM}{\ensuremath{\mathbf{M}}}

\newcommand{\cD}{\ensuremath{\mathcal{D}}}
\newcommand{\cE}{\ensuremath{\mathcal{E}}}

\newcommand{\cH}{\ensuremath{\mathcal{H}}}

\newcommand{\cO}{\ensuremath{\mathcal{O}}}
\newcommand{\cP}{\ensuremath{\mathcal{P}}}
\newcommand{\cQ}{\ensuremath{\mathcal{Q}}}
\newcommand{\cS}{\ensuremath{\mathcal{S}}}
\newcommand{\cU}{\ensuremath{\mathcal{U}}}

\newcommand{\cX}{\ensuremath{\mathcal{X}}}
\newcommand{\cY}{\ensuremath{\mathcal{Y}}}

\global\long\def\one{\scalebox{1.1}{\textnormal{1}} \hspace*{-.75mm} \scalebox{0.7}{\raisebox{.3em}{\bf |}} }

\newcommand{\bZnn}{\ensuremath{\mathbb{Z}_{\geq 0}}}

\newcommand{\ii}{\mathrm{i}\,}

\newcommand{\PowerSet}[1]{2^{#1}}

\newcommand{\Span}{\operatorname{Span}}

\global\long\def\ud{\,\mathrm{d}}

\newcommand{\cGprimal}[1]{\mathcal{G}_{#1}^{\scriptscriptstyle{\color{gray}\bullet}}}
\newcommand{\cVprimal}[1]{\mathcal{V}_{#1}^{\scriptscriptstyle{\color{gray}\bullet}}}
\newcommand{\cEprimal}[1]{\mathcal{E}_{#1}^{\scriptscriptstyle{\color{gray}\bullet}}}
\newcommand{\cFprimal}[1]{\mathcal{F}_{#1}}
\newcommand{\cGdual}[1]{\mathcal{G}_{#1}}
\newcommand{\cEdual}[1]{\mathcal{E}_{#1}}
\newcommand{\Eshort}[1]{\mathscr{E}^{\textup{short}}_{#1}}
\newcommand{\Elong}[1]{\mathscr{E}^{\textup{long}}_{#1}}
\global\long\def\eprimal{e^{\scriptscriptstyle{\color{gray}\bullet}}}
\global\long\def\vprimal{v^{\scriptscriptstyle{\color{gray}\bullet}}}
\global\long\def\wprimal{w^{\scriptscriptstyle{\color{gray}\bullet}}}

\newcommand{\Xprimal}{X^{\scriptscriptstyle{\color{gray}\bullet}}}
\newcommand{\Yprimal}{Y^{\scriptscriptstyle{\color{gray}\bullet}}}
\newcommand{\Aprimal}{A^{\scriptscriptstyle{\color{gray}\bullet}}}

\newcommand{\Hprimal}{H^{\scriptscriptstyle{\color{gray}\bullet}}}

\newcommand{\cGcluster}[1]{\mathscr{G}_{#1}^c}
\newcommand{\cVcluster}[1]{\mathscr{V}_{#1}^c}
\newcommand{\vcluster}[2]{{\halfedge}_{#1}^{#2}}
\newcommand{\cEcluster}[1]{\mathscr{E}_{#1}^c}
\global\long\def\ecluster{e^c}

\global\long\def\edual{e}
\global\long\def\halfedge{h}

\global\long\def\Pbc{P_o}
\global\long\def\Ombc{\Omega_o}

\global\long\def\Interaction{U}
\global\long\def\Potential{\cU}

\global\long\def\bface{f(\bhalfedge)}
\global\long\def\eface{f(\ehalfedge)}
\global\long\def\bedge{e(\bhalfedge)}
\global\long\def\eedge{e(\ehalfedge)}

\global\long\def\etime{n_{\textnormal{out}}}

\global\long\def\bhalfedge{h_{\textnormal{in}}}
\global\long\def\ehalfedge{h_{\textnormal{out}}}

\global\long\def\GrN{\varphi^N}
\global\long\def\GrE{\varphi^E}
\global\long\def\GrS{\varphi^S}
\global\long\def\GrW{\varphi^W}

\global\long\def\diamondplus{\smash{\rotatebox[origin=c]{45}{$\boxtimes$}}}

\global\long\def\degree{\mathrm{d}}

\title{On the spin interface distribution for \\ non-integrable variants of the two-dimensional Ising model}  
\date{}

\author{Rafael L.~Greenblatt\thanks{Department of Mathematics, Universit\`a di Roma ``Tor Vergata'', Italy. \protect\url{greenblatt@mat.uniroma2.it}} \;
and \,
Eveliina Peltola\thanks{Department of Mathematics and Systems Analysis, Aalto University, Finland; \protect\\  and 
Institute for Applied Mathematics, University of Bonn, Germany.   \protect\url{eveliina.peltola@aalto.fi}}
}

\setlist[enumerate]{topsep = 1ex, leftmargin=.8cm, itemsep= -2pt}

\let\OLDthebibliography\thebibliography
\renewcommand\thebibliography[1]{
  \OLDthebibliography{#1}
  \setlength{\parskip}{1pt}
  \setlength{\itemsep}{2pt}
}

\allowdisplaybreaks

\newtheorem{thm}{Theorem}[section]
\newtheorem{cor}[thm]{Corollary}
\newtheorem{conj}[thm]{Conjecture}
\newtheorem{lem}[thm]{Lemma}

\newtheorem{prop}[thm]{Proposition}

\theoremstyle{definition} 
\newtheorem{df}[thm]{Definition}

\newtheorem{remark}[thm]{Remark}

\numberwithin{equation}{section}
\numberwithin{figure}{section}

\begin{document}
\maketitle


\begin{abstract}
\noindent 
We point out that the construction of a martingale observable describing the spin interface of the two-dimensional Ising model extends to a class of non-integrable variants of the two-dimensional Ising model, and express it in terms of Grassmann integrals.  Under a conjecture about the scaling limit of this object, which is similar to some results recently obtained using constructive renormalization group methods, this would imply that the distribution of the interface at criticality has the same scaling limit as in the integrable model: Schramm-Loewner evolution \SLE3.
\end{abstract}

\tableofcontents

\newpage

\bigskip{}
\section{Introduction}

A number of critical discrete interface processes in two dimensions are known to converge to versions of Schramm-Loewner evolution (SLE), a universal family of conformally invariant processes satisfying a domain Markov property~\cite{Schramm:Scaling_limits_of_LERW_and_UST}. 
The most common strategy for proofs of this convergence, beginning with~\cite{Smirnov:Critical_percolation_in_the_plane, LSW:Conformal_invariance_of_planar_LERW_and_UST} 
and implemented for a variety of cases 
in~\cite{Schramm-Sheffield:Harmonic_explorer_and_its_convergence_to_SLE4, 
Smirnov:Towards_conformal_invariance_of_2D_lattice_models, 
Schramm-Sheffield:Contour_lines_of_2D_discrete_GFF, 
CDHKS:Convergence_of_Ising_interfaces_to_SLE, 
Izyurov:Critical_Ising_interfaces_in_multiply_connected_domains} (among others),
centers around a family of  ``observables'' depending on the domain shape and a ``marked'' point on the lattice, 
which are martingales with respect to the natural filtration of the growing interface, 
and admit a conformally invariant scaling limit. 
The martingale property of the observable makes it possible to use it to analyze the Loewner driving function of the interface, 
and one then proceeds to show that convergence of the martingale observable implies convergence of the driving function. 
There are a number of techniques which can then be used to prove convergence of the interface itself. 
This strategy has been rigorously carried out for only a few models ---  
including in particular the nearest-neighbor critical Ising model~\cite{CDHKS:Convergence_of_Ising_interfaces_to_SLE}.

For the spin interface of the critical Ising model, 
the ``usual'' observable was developed from a number of closely related objects known variously as 
spinors, parafermions, or discrete fermions, 
first introduced by Kadanoff and Ceva~\cite{Kadanoff-Ceva:Determination_of_an_operator_algebra_for_the_two-dimensional_Ising_model}, 
subsequently identified with primary fields of conformal field theories describing the scaling limits of this and related models~\cite{Nienhuis-Knops:Spinor_exponents_for_the_two-dimensionalPotts_model, 
Cardy-Riva:Holomorphic_parafermions_in_the_Potts_model_and_SLE, 
CHI:Correlations_of_primary_fields_in_the_critical_planar_Ising_model} --- 
and finally applied to the proof of convergence of the interface process to SLE 
in the works initiated by Smirnov~\cite{Smirnov:Towards_conformal_invariance_of_2D_lattice_models,
Smirnov:Conformal_invariance_in_random_cluster_models1, 
CDHKS:Convergence_of_Ising_interfaces_to_SLE} 
also based on earlier ideas by Aizenman, Cardy, and Kenyon, 
among others~\cite{Aizenman:Geometry_of_critical_percolation_and_conformal_invariance, 
Kenyon:Conformal_invariance_of_domino_tiling, 
Kenyon:Dominos_and_the_Gaussian_free_field,
Rajabpour-Cardy:Discretely_holomorphic_parafermions_in_lattice_ZN_models,
Ikhlef-Cardy:Discretely_holomorphic_parafermions_and_integrable_loop_models}.
The relationship between different forms of the observable is reviewed in~\cite{CCK:Revisiting_the_combinatorics_of_the_2D_Ising_model}.  
The fact that the scaling limit of this observable is conformally invariant follows by identifying it 
with a holomorphic function with specific boundary conditions that characterize it uniquely, 
which in turn follows from discrete holomorphicity properties of the lattice observable~\cite{Mercat:Discrete_Riemann_surfaces_and_the_Ising_model, 
Chelkak-Smirnov:Universality_in_2D_Ising_and_conformal_invariance_of_fermionic_observables}\footnote{In the present work, 
we do not attempt to derive or even introduce discrete Cauchy-Riemann equations or s-holomorphicity, 
since they cannot be expected to hold in general for the models we are considering, 
and (unlike other subjects we review in more detail) their only role is in characterizing the scaling limit in the $\lambda=0$ case. 
In the $\lambda=0$ case we do, however, point out the relation of our conventions to those of~\cite{Smirnov:Conformal_invariance_in_random_cluster_models1,
Chelkak-Smirnov:Universality_in_2D_Ising_and_conformal_invariance_of_fermionic_observables, 
CCK:Revisiting_the_combinatorics_of_the_2D_Ising_model, Chelkak:Ising_model_and_s-embeddings_of_planar_graphs} and related works --- see \cref{rem:sholo}.}.
Furthermore, the discrete interface process naturally has the domain Markov property, 
which can be seen as an explanation for the martingale property of the observable 
and the corresponding property of the limiting SLE process.

Customarily (as in, e.g.,~\cite{CDHKS:Convergence_of_Ising_interfaces_to_SLE, 
Kemppainen-Smirnov:Random_curves_scaling_limits_and_Loewner_evolutions, 
Chelkak:Ising_model_and_s-embeddings_of_planar_graphs}), 
the relationship of convergence of this observable to convergence of the interface is shown by a precompactness argument based ultimately on crossing estimates, 
which in turn are shown to hold for the Ising model using monotonicity properties related to the Fortuin-Kasteleyn-Ginibre (FKG) inequalities.
Proofs along the above lines have been formulated for various forms of the Ising model with critical pair interactions corresponding to the edges of some planar graph 
(with criticality defined in a way associated with a planar embedding of some specific form, 
which also plays a crucial role in the identification of the scaling limit, 
as discussed most fully in~\cite{Chelkak:Ising_model_and_s-embeddings_of_planar_graphs}).

\subsection*{Non-nearest-neighbor Ising models.} 

We consider a version of the Ising model defined by the formal Hamiltonian
\begin{align}
	\Hprimal(\sigma)
	:= - J \sum_{\{\vprimal,\wprimal\} \in \cEprimal{}} \sigma_{\vprimal} \sigma_{\wprimal}
	+
	\lambda \sum_{X \subset \cEprimal{}} \, V(X) 
	\prod_{\{\vprimal,\wprimal\} \in X} \sigma_{\vprimal} \sigma_{\wprimal}
	.
	\label{eq:formal_inter_Ham}
\end{align}
Here, the first term is the nearest-neighbor interaction with coupling constant (strength) $J > 0$, 
and $\lambda \in \bR$ is a parameter controlling the strength of the added non-nearest-neighbor multi-spin interaction 
(fixed but small in absolute value compared to $J$), 
and $\cEprimal{}$ denotes the set of nearest-neighbor pairs $\{\vprimal,\wprimal\}$ (edges) on the square lattice, 
$\sigma_{\vprimal} \in \{\pm 1\}$ are spins on its vertices $\vprimal$,  
and the potential $V(X)$ on subsets of edges is translation-invariant and finite-range in the sense that 
it is supported on sets $X$ with uniformly bounded diameter. 
Then, given an inverse temperature $\beta > 0$, the Gibbs measure associated to the Hamiltonian $\Hprimal$ gives 
weight $\exp(-\beta \Hprimal(\sigma))$ to each spin configuration $\sigma = (\sigma_{\vprimal})$,
thus defining a version of the Ising model (in its spin form).

At least for small $|\lambda|$, it is expected that these models belong to the same universality class as 
the $\lambda=0$ case (which is simply the standard Ising model on the square lattice), 
meaning that they exhibit an order-disorder phase transition at a unique critical temperature $1/\beta_c \in (0,\infty)$, 
and several properties of the critical point are the same.  
It has been shown 
that some correlation functions indeed have the same asymptotic behavior up to 
rescalings~\cite{GGM:The_scaling_limit_of_the_energy_correlations_in_non_integrable_Ising_models, AGG:Non_integrable_Ising_models_in_cylindrical_geometry, AGG:Energy_correlations_of_non_integrable_Ising_models,CGG:In_prep},    
and on physical grounds it seems reasonable to expect that the scaling limit of the spin interfaces is also the same \SLE3 process.
Any proof of this will, however, have a number of differences from the standard case: 
generically, the model defined by \cref{eq:formal_inter_Ham} is not in general expected to have discrete holomorphic observables, its spin interface does not have the domain Markov property (at least not in the same sense as the $\lambda=0$ version),
and even the FKG inequality (see~\cite{Grimmett:Random_cluster_model}) used to obtain crossing estimates holds only in some special cases.

\subsection*{Our goal and motivation.} 

We consider the observable defined in~\cite{Chelkak-Smirnov:Universality_in_2D_Ising_and_conformal_invariance_of_fermionic_observables, 
Izyurov:Critical_Ising_interfaces_in_multiply_connected_domains} 
for the spin (or Peierls) interface of the Ising model on a square lattice in the context of
these non-integrable models\footnote{We shall only consider the spin interface; e.g., we do not consider the Fortuyn-Kasteleyn (or random-cluster) interfaces, 
which for the standard Ising model converge to \SLE{16/3} 
\cite{CDHKS:Convergence_of_Ising_interfaces_to_SLE, 
BPW:On_the_uniqueness_of_global_multiple_SLEs, 
Izyurov:On_multiple_SLE_for_the_FK_Ising_model,
FPW:Connection_probabilities_of_multiple_FK_Ising_interfaces}.
In fact, it is not even clear whether there is a sensible generalization of this interface 
to the class of models we are considering: it is defined in terms of random clusters, 
which are not measurable in terms of the spin configuration, 
but require the introduction of auxiliary random variables.  
This can be done for Hamiltonians consisting only of pair interactions with a fixed sign, 
but even then if the interacting pairs do not form a planar graph, 
it is not clear how to identify the boundary of such a cluster with a curve in the plane.  
For Hamiltonians with terms involving more than two spins, the situation is even more complicated.}.
It is still a martingale with respect to the discrete interface process,  
and we show in \cref{thm:main_thm} that it can be represented in terms of 
correlation functions of an interacting lattice fermionic quantum field theory, 
that is, as a Grassmann (Berezin) integral similar to the path integral expression for the Schwinger functions.  
Previously such representations have been given for various observables in periodic and open boundary conditions (based on the high-temperature contour expansion) \cite{GGM:The_scaling_limit_of_the_energy_correlations_in_non_integrable_Ising_models,AGG:Non_integrable_Ising_models_in_cylindrical_geometry}; studying spin interfaces requires a different representation, which however has all of the properties used in a constructive renormalization group (RG) treatment (one should compare our \cref{prop:interacting_Grassmann_main}  
with e.g.~\cite[Proposition~1]{GGM:The_scaling_limit_of_the_energy_correlations_in_non_integrable_Ising_models} or~\cite[Proposition~3.1]{AGG:Non_integrable_Ising_models_in_cylindrical_geometry}) and is somewhat simpler to derive.

Importantly, the Grassmann observable in question is \emph{local} 
(like the energy observable studied in~\cite{GGM:The_scaling_limit_of_the_energy_correlations_in_non_integrable_Ising_models, AGG:Energy_correlations_of_non_integrable_Ising_models}, 
or the spin on the boundary~\cite{Cava:PhD_thesis, CGG:In_prep}, but unlike the generic spin).
Therefore, although we will not attempt to make rigorous conclusions concerning the scaling limit in the present article, 
there are reasons to believe that a scaling limit could be eventually  
obtained by constructive RG methods. 
The purpose of this article is to offer a \emph{specific setup and strategy} to build on recent developments on the RG side combined with the already established non-interacting case,
for deriving the scaling limit up to specific technical issues that we discuss in \Cref{sec:conclusions}.
Our main aim is to clarify exactly what needs to be proven by RG methods in order to obtain convergence of interfaces, intended to prepare the way for future work in this direction.
This strategy would lead to a proof that the observable is simply rescaled 
(along with a change in the critical temperature, cf. \cref{conj:interacting_convergence}), 
leading to \emph{exactly the same}\footnote{This should be distinguished from, e.g., interacting dimer models~\cite{GMT:Haldane_relation_for_interacting_dimers}, 
where at least some observables change more qualitatively but are still expected to have conformally invariant scaling limits.} 
scaling limit for the interface at the appropriate critical temperature as in the integrable case: \SLE3.

\subsection*{Additional remarks.}

Note that convergence of the observable might, a priori, not be sufficient to guarantee convergence of the interfaces in a strong sense in a space of curves.
As already mentioned, in existing proofs for the Ising model, this implication is obtained using crossing estimates~\cite{Aizenman-Burchard:Holder_regularity_and_dimension_bounds_for_random_curves,Kemppainen-Smirnov:Random_curves_scaling_limits_and_Loewner_evolutions}, 
which are ultimately based on the FKG inequality --- which does not generally hold in the non-integrable case. 
Nevertheless, by quite general arguments 
as in the original approach of Lawler, Schramm, and Werner~\cite{LSW:Conformal_invariance_of_planar_LERW_and_UST}, 
convergence of a martingale observable in a suitable locally uniform sense does imply that the \emph{Loewner driving function} of the interface converges. 
Though this is a priori much weaker than convergence of the interface process itself, 
under a \emph{reversibility condition}, a result of Sheffield and Sun~\cite{Sheffield-Sun:Strong_path_convergence_from_Loewner_driving_function_convergence} 
(originating from an earlier project of Schramm and Sheffield)
shows that it actually implies convergence of the interface.
The relevant reversibility property does not hold for all interface processes, 
but it does hold for our specific choice of the definition of the interface as a function of the low-temperature contour representation (\cref{def:interface}).

One feature of our proposal --- which may be surprising ---
is that the discrete interface process does not enjoy the domain Markov property, which is one of the defining features of \SLE3.  
In the standard Ising model on a planar graph, 
conditioned on the presence of a Peierls interface dividing the system into two parts, 
everything else about these two parts is (conditionally) independent. 
This is not the case in general, since the Hamiltonian~\eqref{eq:formal_inter_Ham} can {(and will) 
include interaction terms crossing the interface. 
Consequently, the distribution of the interface conditioned on an initial segment is not measurable with respect to the \emph{hull} of that segment, as the domain Markov property would require.
However, since the observable is still a martingale with respect to the interface process, this is not a problem for demonstrating convergence of the interface,
and what we propose in the present work is that the conjectured behavior of the observables 
still makes it possible to conclude that the full domain Markov property 
is restored as an \emph{emergent property} of the scaling limit, 
similar to conformal invariance (which also does not have an exact discrete analogue in all of the models under consideration).

Related to this, let us also note that when defining boundary conditions, 
it is necessary to specify the spin configuration at some distance from the boundary ---
in particular, it takes more than labelling two points on the boundary 
to specify ``Dobrushin'' type boundary conditions. 
Part of \cref{conj:interacting_convergence} is a bound on the effect of this additional information on the observable, which would imply that the additional boundary condition has no influence on the scaling limit, restoring the expected measurability property for the continuous observable and thus for the conditional distribution of the continuous interface.

Let us finally note that our work is limited to a square lattice, 
whereas for the planar nearest-neighbor Ising model ($\lambda=0$)
the spin interface is known to converge to SLE$(3)$ for a very general class of graphs,
as systematically developed in particular by Chelkak 
\cite{Chelkak-Smirnov:Universality_in_2D_Ising_and_conformal_invariance_of_fermionic_observables, 
Chelkak:Ising_model_and_s-embeddings_of_planar_graphs}.
For non-integrable interactions, the situation is likely to be significantly more complex for the following reason. 
The interaction in~\eqref{eq:formal_inter_Ham} has an effect analogous to changing the coupling of the planar model (thus changing the critical temperature), 
a priori changing each edge in a way depending on the structure the Hamiltonian. 
For the square lattice it is possible to guarantee that this effect of the interaction is uniform 
(apart from boundary effects) by stipulating that it respects the symmetries of $\bZ^2$ sufficiently well. 
In many other cases, the simplest plausible scenario is that the long-distance behavior of the system resembles 
that of an Ising model with inhomogeneous couplings, which includes cases where the critical behavior of the model is dramatically 
changed\footnote{E.g., disordered models which are translation-invariant in one direction apparently have a critical point 
only in a very different sense than the familiar case~\cite{McCoy-Wu:2D_Ising_model, Fisher:Critical_behavior_of_random_transverse_field_Ising_spin_chains} (see also~\cite{CGGContinuum_limit_of_random_matrix_products_in_statistical_mechanics_of_disordered_systems, Giacomin-Greenblatt:Lyapunov_exponent_for_products_of_random_Ising_transfer_matrices} 
for recent progress towards a rigorous proof of this statement), 
and hence should not exhibit conformal invariance.}.

\subsection*{Outline of the rest of this article.}

In \cref{sec:prelim}, we review the definition of the Ising model with Dobrushin boundary conditions, the low-temperature contour representation, and the spin (or Peierls) interface. 
We do this partly to make the present work as self-contained as possible, but more so to introduce many definitions used in the sequel
(some of which involve varying conventions in the literature), 
as well as to emphasize some features of these objects particularly relevant to the present work. 
In particular, in \cref{sec:interface} we introduce an  interface exploration process (\cref{def:interface})
and 
the associated martingale observable (\cref{prop:M_general}), 
which is a variant of the one used to prove the conformal invariance of the critical Ising model in the scaling limit in Smirnov's seminal works~\cite{Smirnov:Towards_conformal_invariance_of_2D_lattice_models}. 
In \cref{sec:Grassman}, we review a representation of the planar Ising model in terms of Grassmann variables and give an expression for the martingale observable in this framework, 
i.e., as a Grassmann (Berezin) integral, similar to~\cite{CCK:Revisiting_the_combinatorics_of_the_2D_Ising_model,Samuel:The_use_of_anticommuting_variable_integrals_in_statistical_mechanics}.
In particular, we do this using the relationship to the low-temperature contour representation in a way which prepares the later generalization.

Then, in \cref{sec:interact} we arrive at the modified (non-integrable) model, 
which we examine first in terms of low-temperature contours and then use them 
to derive a Grassmann representation, including an expression for a martingale observable 
(\cref{eq:M_interacting} and \cref{thm:main_thm}).  
This should make it possible to study the scaling limit via rigorous RG methods.
Unfortunately, thus far these methods have only been implemented for toroidal and cylindrical domains --- 
so we are not able to proceed any further in a rigorous fashion at this stage.  
The results obtained thus far do, however, lead to the specific, natural \cref{conj:interacting_convergence} about these correlation functions, 
which would imply their convergence in a locally uniform fashion to the same scaling limit as in the planar case. 
In \cref{sec:conjecture}, we present this conjecture and explain how it implies the convergence of the interface to \SLE3.
We end this article with some further comments in \cref{sec:conclusions}.

\subsection*{Acknowledgments.}

We thank Hugo Duminil-Copin and R\'emy Mahfouf for insightful correspondences.

The work of R.L.G.~was supported by the European Research Council (ERC) under the European Union's Horizon 2020 research and innovation programme (ERC StG MaMBoQ, grant agreement No.\ 802901 for R.L.G.) and by the MIUR Excellence Department Project MatMod@TOV awarded to the Department of Mathematics, University of Rome Tor Vergata, CUP E83C18000100006.

E.P.~was supported 
by the Deutsche Forschungsgemeinschaft (DFG, German Research Foundation) under Germany's Excellence Strategy {EXC-2047/1-390685813}, 
the DFG collaborative research centre ``The mathematics of emerging effects'' CRC-1060/211504053,
as well as by 
the Academy of Finland Centre of Excellence Programme grant number 346315 ``Finnish centre of excellence in Randomness and STructures (FiRST)'',  
the Academy of Finland grant number 340461 ``Conformal invariance in planar random geometry'', 
and the European Research Council (ERC) under the European Union's Horizon 2020 research and innovation programme (101042460): 
ERC Starting grant ``Interplay of structures in conformal and universal random geometry'' (ISCoURaGe).

%
%
%

\bigskip{}
\section{Preliminaries on the nearest-neighbor Ising model}
\label{sec:prelim}

In this section we present a number of basic objects underlying our analysis.
While the main ideas are already well-known, we take the occasion to set up important conventions for later use and highlight some features which are particularly important here. 

In \cref{sec:setup} we summarize the geometrical objects underlying all other definitions.
Then in \cref{sec:Ising}, we define the nearest-neighbor (spin-)Ising model, its partition functions, and the boundary conditions of interest. 
The low-temperature expansion of the domain walls between the variable spins gives rise to a contour representation of the Ising model, i.e.\ a formulation of the underlying probability space in terms of geometric objects rather than spins. 
We then define various objects related to the interface in terms of such a model in \cref{sec:interface}, in particular
the associated exploration process (\cref{def:interface}) and martingale observable (\cref{prop:M_general}).
Crucially, the definitions --- and certain key properties --- do not depend on which probability measure is used, and so they generalize immediately to the non-integrable systems which are the main subject of this article.

\begin{figure}[t]
\centering
\begin{subfigure}[t]{0.45\textwidth}
\centering
\tikzpicturedependsonfile{graph_parts.tikz}
\begin{tikzpicture}
	\draw (0,0) -- (0,1);
	\draw (0,0) -- (1,0);
	\draw (0,1) -- (0,2);
	\draw (0,1) -- (1,1);
	\draw (0,2) -- (0,3);
	\draw (0,2) -- (1,2);
	\draw (0,3) -- (1,3);
	\draw (1,0) -- (1,1);
	\draw (1,1) -- (1,2);
	\draw (1,1) -- (2,1);
	\draw (1,2) -- (1,3);
	\draw (2,1) -- (3,1);
	\draw (3,0) -- (3,1);
	\draw (3,1) -- (3,2);
	\draw (-1,2) -- (0,2);
	\draw (0,0) -- (-1,0);
	\draw (0,0) -- (0,-1);
	\draw (0,1) -- (-1,1);
	\draw (0,3) -- (-1,3);
	\draw (0,3) -- (0,4);
	\draw (1,0) -- (2,0);
	\draw (1,0) -- (1,-1);
	\draw (1,2) -- (2,2);
	\draw (1,3) -- (1,4);
	\draw (1,3) -- (2,3);
	\draw (2,1) -- (2,2);
	\draw (2,1) -- (2,0);
	\draw (3,0) -- (3,-1);
	\draw (3,0) -- (2,0);
	\draw (3,0) -- (4,0);
	\draw (3,1) -- (4,1);
	\draw (3,2) -- (3,3);
	\draw (3,2) -- (2,2);
	\draw (3,2) -- (4,2);
	\draw (-1,2) -- (-1,1);
	\draw (-1,2) -- (-1,3);
	\draw (-1,2) -- (-2,2);
	\draw[dualEdge] (3.5,0.5) -- (3.5,1.5);
	\draw[dualEdge] (3.5,-0.5) -- (3.5,0.5);
	\draw[dualEdge] (0.5,2.5) -- (0.5,3.5);
	\draw[dualEdge] (0.5,2.5) -- (1.5,2.5);
	\draw[dualEdge] (2.5,1.5) -- (2.5,2.5);
	\draw[dualEdge] (2.5,1.5) -- (3.5,1.5);
	\draw[dualEdge] (0.5,3.5) -- (1.5,3.5);
	\draw[dualEdge] (-0.5,3.5) -- (0.5,3.5);
	\draw[dualEdge] (1.5,2.5) -- (1.5,3.5);
	\draw[dualEdge] (2.5,2.5) -- (3.5,2.5);
	\draw[dualEdge] (1.5,1.5) -- (2.5,1.5);
	\draw[dualEdge] (1.5,1.5) -- (1.5,2.5);
	\draw[dualEdge] (-0.5,-0.5) -- (-0.5,0.5);
	\draw[dualEdge] (-0.5,-0.5) -- (0.5,-0.5);
	\draw[dualEdge] (-0.5,0.5) -- (0.5,0.5);
	\draw[dualEdge] (-0.5,0.5) -- (-0.5,1.5);
	\draw[dualEdge] (0.5,0.5) -- (1.5,0.5);
	\draw[dualEdge] (0.5,0.5) -- (0.5,1.5);
	\draw[dualEdge] (0.5,-0.5) -- (0.5,0.5);
	\draw[dualEdge] (0.5,-0.5) -- (1.5,-0.5);
	\draw[dualEdge] (-1.5,2.5) -- (-0.5,2.5);
	\draw[dualEdge] (1.5,0.5) -- (1.5,1.5);
	\draw[dualEdge] (1.5,0.5) -- (2.5,0.5);
	\draw[dualEdge] (1.5,-0.5) -- (1.5,0.5);
	\draw[dualEdge] (0.5,1.5) -- (0.5,2.5);
	\draw[dualEdge] (0.5,1.5) -- (1.5,1.5);
	\draw[dualEdge] (-0.5,1.5) -- (0.5,1.5);
	\draw[dualEdge] (-0.5,1.5) -- (-0.5,2.5);
	\draw[dualEdge] (3.5,1.5) -- (3.5,2.5);
	\draw[dualEdge] (2.5,0.5) -- (3.5,0.5);
	\draw[dualEdge] (2.5,0.5) -- (2.5,1.5);
	\draw[dualEdge] (2.5,-0.5) -- (3.5,-0.5);
	\draw[dualEdge] (2.5,-0.5) -- (2.5,0.5);
	\draw[dualEdge] (-0.5,2.5) -- (0.5,2.5);
	\draw[dualEdge] (-0.5,2.5) -- (-0.5,3.5);
	\draw[dualEdge] (-1.5,1.5) -- (-1.5,2.5);
	\draw[dualEdge] (-1.5,1.5) -- (-0.5,1.5);
	\draw (1.5,3.5) node[face] {};
	\draw (3.5,0.5) node[face] {};
	\draw (3.5,-0.5) node[face] {};
	\draw (0.5,2.5) node[face] {};
	\draw (2.5,1.5) node[face] {};
	\draw (0.5,3.5) node[face] {};
	\draw (-0.5,3.5) node[face] {};
	\draw (1.5,2.5) node[face] {};
	\draw (2.5,2.5) node[face] {};
	\draw (1.5,1.5) node[face] {};
	\draw (3.5,2.5) node[face] {};
	\draw (-0.5,-0.5) node[face] {};
	\draw (-0.5,0.5) node[face] {};
	\draw (0.5,0.5) node[face] {};
	\draw (0.5,-0.5) node[face] {};
	\draw (-1.5,2.5) node[face] {};
	\draw (1.5,0.5) node[face] {};
	\draw (1.5,-0.5) node[face] {};
	\draw (0.5,1.5) node[face] {};
	\draw (-0.5,1.5) node[face] {};
	\draw (3.5,1.5) node[face] {};
	\draw (2.5,0.5) node[face] {};
	\draw (2.5,-0.5) node[face] {};
	\draw (-0.5,2.5) node[face] {};
	\draw (-1.5,1.5) node[face] {};
	\draw (0,0) node[vertex] {};
	\draw (0,1) node[vertex] {};
	\draw (0,2) node[vertex] {};
	\draw (0,3) node[vertex] {};
	\draw (1,0) node[vertex] {};
	\draw (1,1) node[vertex] {};
	\draw (1,2) node[vertex] {};
	\draw (1,3) node[vertex] {};
	\draw (2,1) node[vertex] {};
	\draw (3,0) node[vertex] {};
	\draw (3,1) node[vertex] {};
	\draw (3,2) node[vertex] {};
	\draw (-1,2) node[vertex] {};
	\draw (-1,1) node[edgeVertex] {};
	\draw (4,1) node[edgeVertex] {};
	\draw (-1,0) node[edgeVertex] {};
	\draw (3,3) node[edgeVertex] {};
	\draw (-1,3) node[edgeVertex] {};
	\draw (2,2) node[edgeVertex] {};
	\draw (-2,2) node[edgeVertex] {};
	\draw (1,4) node[edgeVertex] {};
	\draw (3,-1) node[edgeVertex] {};
	\draw (2,0) node[edgeVertex] {};
	\draw (2,3) node[edgeVertex] {};
	\draw (0,4) node[edgeVertex] {};
	\draw (4,2) node[edgeVertex] {};
	\draw (0,-1) node[edgeVertex] {};
	\draw (1,-1) node[edgeVertex] {};
	\draw (4,0) node[edgeVertex] {};
\end{tikzpicture}
\end{subfigure}
~
\begin{subfigure}[t]{0.45\textwidth}
\centering
\tikzpicturedependsonfile{Kasteleyn_graph_parts.tikz}
\begin{tikzpicture}
	\draw (-1,1) node[vertex,draw=white,fill=white] {};
	\draw (4,1) node[vertex,draw=white,fill=white] {};
	\draw (-1,0) node[vertex,draw=white,fill=white] {};
	\draw (3,3) node[vertex,draw=white,fill=white] {};
	\draw (-1,3) node[vertex,draw=white,fill=white] {};
	\draw (2,2) node[vertex,draw=white,fill=white] {};
	\draw (-2,2) node[vertex,draw=white,fill=white] {};
	\draw (1,4) node[vertex,draw=white,fill=white] {};
	\draw (3,-1) node[vertex,draw=white,fill=white] {};
	\draw (2,0) node[vertex,draw=white,fill=white] {};
	\draw (2,3) node[vertex,draw=white,fill=white] {};
	\draw (0,4) node[vertex,draw=white,fill=white] {};
	\draw (4,2) node[vertex,draw=white,fill=white] {};
	\draw (0,-1) node[vertex,draw=white,fill=white] {};
	\draw (1,-1) node[vertex,draw=white,fill=white] {};
	\draw (4,0) node[vertex,draw=white,fill=white] {};
	\draw[clusterLongEdge] (3.5,0.5) -- (3.5,1.5);
	\draw[clusterLongEdge] (3.5,-0.5) -- (3.5,0.5);
	\draw[clusterLongEdge] (0.5,2.5) -- (0.5,3.5);
	\draw[clusterLongEdge] (0.5,2.5) -- (1.5,2.5);
	\draw[clusterLongEdge] (2.5,1.5) -- (2.5,2.5);
	\draw[clusterLongEdge] (2.5,1.5) -- (3.5,1.5);
	\draw[clusterLongEdge] (0.5,3.5) -- (1.5,3.5);
	\draw[clusterLongEdge] (-0.5,3.5) -- (0.5,3.5);
	\draw[clusterLongEdge] (1.5,2.5) -- (1.5,3.5);
	\draw[clusterLongEdge] (2.5,2.5) -- (3.5,2.5);
	\draw[clusterLongEdge] (1.5,1.5) -- (2.5,1.5);
	\draw[clusterLongEdge] (1.5,1.5) -- (1.5,2.5);
	\draw[clusterLongEdge] (-0.5,-0.5) -- (-0.5,0.5);
	\draw[clusterLongEdge] (-0.5,-0.5) -- (0.5,-0.5);
	\draw[clusterLongEdge] (-0.5,0.5) -- (0.5,0.5);
	\draw[clusterLongEdge] (-0.5,0.5) -- (-0.5,1.5);
	\draw[clusterLongEdge] (0.5,0.5) -- (1.5,0.5);
	\draw[clusterLongEdge] (0.5,0.5) -- (0.5,1.5);
	\draw[clusterLongEdge] (0.5,-0.5) -- (0.5,0.5);
	\draw[clusterLongEdge] (0.5,-0.5) -- (1.5,-0.5);
	\draw[clusterLongEdge] (-1.5,2.5) -- (-0.5,2.5);
	\draw[clusterLongEdge] (1.5,0.5) -- (1.5,1.5);
	\draw[clusterLongEdge] (1.5,0.5) -- (2.5,0.5);
	\draw[clusterLongEdge] (1.5,-0.5) -- (1.5,0.5);
	\draw[clusterLongEdge] (0.5,1.5) -- (0.5,2.5);
	\draw[clusterLongEdge] (0.5,1.5) -- (1.5,1.5);
	\draw[clusterLongEdge] (-0.5,1.5) -- (0.5,1.5);
	\draw[clusterLongEdge] (-0.5,1.5) -- (-0.5,2.5);
	\draw[clusterLongEdge] (3.5,1.5) -- (3.5,2.5);
	\draw[clusterLongEdge] (2.5,0.5) -- (3.5,0.5);
	\draw[clusterLongEdge] (2.5,0.5) -- (2.5,1.5);
	\draw[clusterLongEdge] (2.5,-0.5) -- (3.5,-0.5);
	\draw[clusterLongEdge] (2.5,-0.5) -- (2.5,0.5);
	\draw[clusterLongEdge] (-0.5,2.5) -- (0.5,2.5);
	\draw[clusterLongEdge] (-0.5,2.5) -- (-0.5,3.5);
	\draw[clusterLongEdge] (-1.5,1.5) -- (-1.5,2.5);
	\draw[clusterLongEdge] (-1.5,1.5) -- (-0.5,1.5);
	\draw (1.5,3.5) node[clusterFace] {};
	\draw (3.5,0.5) node[clusterFace] {};
	\draw (3.5,-0.5) node[clusterFace] {};
	\draw (0.5,2.5) node[clusterFace] {};
	\draw (2.5,1.5) node[clusterFace] {};
	\draw (0.5,3.5) node[clusterFace] {};
	\draw (-0.5,3.5) node[clusterFace] {};
	\draw (1.5,2.5) node[clusterFace] {};
	\draw (2.5,2.5) node[clusterFace] {};
	\draw (1.5,1.5) node[clusterFace] {};
	\draw (3.5,2.5) node[clusterFace] {};
	\draw (-0.5,-0.5) node[clusterFace] {};
	\draw (-0.5,0.5) node[clusterFace] {};
	\draw (0.5,0.5) node[clusterFace] {};
	\draw (0.5,-0.5) node[clusterFace] {};
	\draw (-1.5,2.5) node[clusterFace] {};
	\draw (1.5,0.5) node[clusterFace] {};
	\draw (1.5,-0.5) node[clusterFace] {};
	\draw (0.5,1.5) node[clusterFace] {};
	\draw (-0.5,1.5) node[clusterFace] {};
	\draw (3.5,1.5) node[clusterFace] {};
	\draw (2.5,0.5) node[clusterFace] {};
	\draw (2.5,-0.5) node[clusterFace] {};
	\draw (-0.5,2.5) node[clusterFace] {};
	\draw (-1.5,1.5) node[clusterFace] {};
\end{tikzpicture}
\end{subfigure}
\caption{An example of the graphs being used.
On the left,
the points of $\Omega \subset \cVprimal\Omega$ (interior primal vertices), 
where the variable spins will be located, are shown as ``\raisebox{-1mm}{\tikzvertex{vertex}}'',  
the boundary vertices in $\partial\cVprimal\Omega$ (having fixed spins)  
as ``\raisebox{-1mm}{\tikzvertex{edgeVertex}}'', and edges in $\cEprimal\Omega$ as solid lines.  
The faces $\cFprimal\Omega$ (identified with dual vertices) are drawn as ``\raisebox{-1mm}{\tikzvertex{face}}'', 
and the associated dual edges $\cEdual\Omega$ as dashed lines.
On the right, the corresponding cluster graph $\cGcluster\Omega = ( \cVcluster\Omega, \cEcluster\Omega )$.
}
\label{fig:Omega_example}
\end{figure}

\subsection{Primal, dual, and cluster graphs and contour configurations}
\label{sec:setup}

We consider spin models on subsets of the two-dimensional 
square lattice $\bZ^2 = (\cVprimal{},\cEprimal{})$, whose vertices we identify with elements of 
$\cVprimal{} := \bZ + \ii \bZ$ (the set of complex numbers with integer real and imaginary parts)
and edges comprise the set $\cEprimal{}$ of unordered pairs 
$\{ \vprimal, \wprimal \} \subset \bZ + \ii \bZ$ with $|\vprimal - \wprimal |=1$.
The following definitions are illustrated in \cref{fig:Omega_example}.

Let $\Omega \subset \cVprimal{}$ be a finite subset of the vertices and 
let $\cEprimal\Omega \subset \cEprimal{}$ be the set of edges with at least one element in $\Omega$.
Let $\cVprimal\Omega \supset \Omega$ be the set of vertices in $\cVprimal{}$ which are at distance one or less from $\Omega$, and set $\partial \cVprimal\Omega := \cVprimal\Omega \setminus \Omega$.
We call $\cGprimal\Omega = (\cVprimal\Omega,\cEprimal\Omega)$ the \emph{primal graph}.

Next, let $\cFprimal{}$ be the set of faces of the lattice $\bZ^2$, which are unit squares, naturally identified with points in $(\bZ + \tfrac12) + \ii ( \bZ + \tfrac12 )$ (vertices of the dual lattice).  
Let $\cEdual{}$ be the set of unordered pairs of adjacent faces in $\cFprimal{}$, i.e.,~dual edges.
Note that each edge $\eprimal \in \cEprimal{}$ crosses exactly one element of $\cEdual{}$, which we denote $\edual$\footnote{Note also that faces of the dual lattice naturally correspond to (primal) vertices $\cVprimal{}$.}; 
we will frequently use this one-to-one correspondence $\eprimal \leftrightarrow \edual$.  
We say that $\edual \in \cEdual{}$ is \emph{incident to} a face $f \in \cFprimal{}$, and write $\edual \sim f$, 
if the corresponding dual vertex at $f$ is one of the endpoints of the dual edge $\edual$.

Let $\cEdual\Omega$ be the set of dual edges $\edual$ associated with edges $\eprimal \in \cEprimal\Omega$, 
and let $\cFprimal\Omega$ be the set of the faces incident to the dual edges in $\cEdual\Omega$.
We call $\cGdual\Omega = (\cFprimal\Omega,\cEdual\Omega)$ the \emph{dual graph}.
We~consider contour configurations on the dual graph.

We will make extensive use of the representation of the Ising model in terms of Grassmann variables 
indexed by ``half-edges'', which are pairs $h = (f(h),e(h)) \in \cFprimal{} \times \cEdual{}$.  
Let $\cVcluster\Omega$ be the set of half-edges with $f(h) \in \cFprimal\Omega$ (but not necessarily $e(h) \in \cEdual\Omega$).
Let $\Eshort\Omega$ (so-called ``short edges'') be the set of unordered pairs $\{\halfedge_1,\halfedge_2\} \subset \cVcluster\Omega$ with $f(\halfedge_1) = f(\halfedge_2)$ and $e(\halfedge_1) \neq e(\halfedge_2)$, 
and let $\Elong\Omega$ be the set of pairs with $e(\halfedge_1) = e(\halfedge_2) \in \cEdual\Omega$ and $f(\halfedge_1) \neq f(\halfedge_2)$.
The graph $\cGcluster\Omega = (\cVcluster\Omega, \cEcluster\Omega)$, 
where $\cEcluster\Omega := \Eshort\Omega \cup \Elong\Omega$, 
is a version of the \emph{cluster graph} or \emph{terminal graph} introduced by Kasteleyn~\cite{Kasteleyn:Dimer_statistics_and_phase_transitions} and illustrated in \cref{fig:Omega_example}.
Compared with~\cite{CCK:Revisiting_the_combinatorics_of_the_2D_Ising_model}, for example, our definition has some apparently redundant elements, which we find useful for introducing boundary conditions and for identifying elements of the graphs associated with different domains $\Omega$.
Note that this graph, unlike the others introduced so far, is necessarily nonplanar.  
The set $\Elong\Omega$ of long edges is naturally in one-to-one correspondence with $\cEdual\Omega$, 
and for each face $f \in \cFprimal\Omega$ there are exactly four elements of $\Eshort\Omega$ which form the edges of a complete 4-graph, which we denote by $\diamondplus_f$.

Kasteleyn introduced the cluster graph to relate the Ising model to a dimer model  in a way that is closely related to the representation we will use in this article.
A \emph{perfect matching} (dimer configuration) 
on $\cGcluster\Omega$ is a collection $\{ \ecluster_1, \ecluster_2, \ldots, \ecluster_{n} \} \subset \cEcluster\Omega$ of $n = \tfrac12 |\cVcluster\Omega|$ edges in which every vertex of $\cGcluster\Omega$ appears exactly once.
Let $\cP_\Omega$ be the set of even-degree subgraphs  
of the dual graph, i.e., \emph{contour configurations} $P \subset \cEdual\Omega$ 
on dual edges 
for which the \emph{degree} $\degree_f(P) := \big|\big\{ \edual \in P \;|\; \edual \sim f \big\}\big|$ 
is even for all $f \in \cFprimal\Omega$. 
On the cluster graph $\cGcluster\Omega$, any perfect matching always includes an even number of long edges incident to each face $f \in \cFprimal\Omega$. 
Hence, it can be put into (nonunique) correspondence with a contour configuration $P \in \cP_\Omega$, by completing the perfect matching at the short edges 
(note that there can be several ways to match the short edges that results in the same contour configuration). 
On the other hand, as long as one spin on the boundary is fixed spin configurations $\sigma$ on $\Omega$   
correspond bijectively to contour configurations $P \in \cP_\Omega$ (see \cref{sec:Ising}).
This gives a correspondence of spin configurations on $\Omega$ and perfect matchings (dimer covers) on $\cGcluster\Omega$, which makes it possible to express many quantities of interest in terms of Pfaffians of antisymmetric matrices related to the adjacency matrix of $\cGcluster\Omega$~\cite{Kasteleyn:Dimer_statistics_and_phase_transitions}.

\subsection{Ising model and its low-temperature contour representation}
\label{sec:Ising}

Standard references for the Ising model include the books~\cite{McCoy-Wu:2D_Ising_model, Baxter:Exactly_solved_models_in_statistical_mechanics}. 
A \emph{spin configuration} on $\Omega$ is an assignment 
$\sigma \in \{\pm1\}^{\cVprimal\Omega}$ of variable spins $\sigma_{\vprimal}$ equaling $-1$ or $+1$ at each vertex $\vprimal \in \cVprimal\Omega$. 
The \emph{Ising model} is a Boltzmann distribution on 
the set of all spin configurations. 
Given coupling constants $\ul J = (J_{\eprimal})_{\eprimal \in \cEprimal\Omega} \in [0,\infty)^{\cEprimal\Omega}$, the (ferromagnetic, nearest-neighbor) Ising Hamiltonian with ``$+$'' boundary conditions is\footnote{This differs from the most commonly used Ising spin Hamiltonian by a constant term $-\sum_{\eprimal \in \cEprimal\Omega} J_{\eprimal}$.} 
\begin{align}
H_{\Omega,\ul J}^{+}(\sigma) 
=
2 \sum_{\eprimal \in \cEprimal\Omega}
J_{\eprimal} \, \epsilon_{\eprimal}(\sigma) ,
\label{eq:Ham_plus} 
\end{align}
where $\epsilon_{\eprimal}(\sigma) = \tfrac12 (1 - \sigma_{\vprimal} \sigma_{\wprimal})$ for edges $\eprimal = \{ \vprimal,\wprimal \}$, with $\sigma_{\vprimal} = 1$ for all $\vprimal \in \partial \cVprimal\Omega$.
Let $\beta \in (0,\infty)$ be the inverse temperature and $\ul x = (x_{\eprimal})_{\eprimal \in \cEprimal\Omega}$ the collection of
\emph{edge weights} 
\begin{align}
\label{eq:edge_weights}
x_{\edual} = \exp(-2 \beta J_{\eprimal})  
\end{align}
where $\edual$ is the dual edge crossing $\eprimal$.
As $\epsilon_{\eprimal} = \epsilon_{\eprimal}(\sigma) \in \{ 0,1 \}$ for all $\sigma$, we can write
\begin{align}
\exp \big( - \beta H_{\Omega,\ul J}^{+}(\sigma) \big)
=
\prod_{\eprimal \in \cEprimal\Omega} x_{\edual}^{\epsilon_{\eprimal}(\sigma)}
=
\prod_{\substack{\eprimal \in \cEprimal\Omega \\ \epsilon_{\eprimal}(\sigma) = 1}} x_{\edual}
=
\prod_{\edual \in P(\sigma)} x_{\edual} ,
\label{eq:Gibbs_weight_plus}
\end{align}
using the one-to-one correspondence $\eprimal \leftrightarrow \edual$ 
and writing $P(\sigma) := \big\{ \edual \in \cEdual\Omega \;|\; \epsilon_{\eprimal} (\sigma) = 1 \big\}$.
The range of $P$ is exactly the $\cP_\Omega \subset 2^{\cEdual\Omega}$ which are \emph{even} in the sense that
$\degree_f(P)$ is even for all $f \in \cFprimal\Omega$.
In fact each element of $\cP_\Omega$ is associated with a unique spin configuration, so the Ising model is equivalent to one on this state space, known as the \emph{low-temperature contour representation}.
We  have the partition function
\begin{align}
Z_{\Omega,\ul x}^{+}
:=
\sum_{\sigma \in \{ \pm1 \}^\Omega} e^{- \beta H_{\Omega,\ul J}^{+}(\sigma)}
=
\sum_{P \in \cP_\Omega} \prod_{\edual \in P} x_{\edual} .
\label{eq:Z_plus_lt}
\end{align}
and the distribution of the random variable $P(e)$ is the measure
\begin{align*}
\bP^{+}_{\Omega,\ul x}[P] := \frac{1}{Z^{+}_{\Omega,\ul x}} \, \prod_{\edual \in P} x_{\edual} , \qquad P \in \cP_\Omega ,
\end{align*}
which we consider as being parameterized by the edge weights $\ul x$.

This can be generalized to Dobrushin boundary conditions as follows.  
For fixed $\xi = (\bhalfedge,\ehalfedge) \in \cVcluster\Omega \times \cVcluster\Omega$ with $\bhalfedge \neq \ehalfedge$, let $\cP_\Omega^\xi$ be the set of $P \subset \cEdual\Omega \setminus \{\bedge,\eedge\}$ such that $P$ is even with the two indicated half-edges counted in the degree of their face, i.e., 
\begin{align}
\nonumber
\degree_f(P) 
+ \one_{\bface}(f) + \one_{\eface} 
= \; &
\big| \big\{ \edual \in P \;|\; \edual \sim f \big\} \big|
+ \one_{\bface}(f) + \one_{\eface} 
\\
\label{eq:contour_collection_requirement}
\in \; & 2 \bZ , 
\qquad \textnormal{for all } f \in \cFprimal\Omega .
\end{align}
See \cref{fig:P_example} for an example.
We also write $\cP^{\emptyset}_\Omega := \cP_\Omega$ for the earlier even-degree contours.

\begin{figure}[h!]
\centering
\tikzpicturedependsonfile{P_example.tikz} \begin{tikzpicture}
	\draw[dualEdge] (3.5,0.5) -- (3.5,1.5);
	\draw[dualEdge] (3.5,-0.5) -- (3.5,0.5);
	\draw[dualEdge] (0.5,2.5) -- (0.5,3.5);
	\draw[dualEdge] (0.5,2.5) -- (1.5,2.5);
	\draw[dualEdge] (2.5,1.5) -- (2.5,2.5);
	\draw[dualEdge] (2.5,1.5) -- (3.5,1.5);
	\draw[dualEdge] (0.5,3.5) -- (1.5,3.5);
	\draw[dualEdge] (-0.5,3.5) -- (0.5,3.5);
	\draw[dualEdge] (1.5,2.5) -- (1.5,3.5);
	\draw[dualEdge] (2.5,2.5) -- (3.5,2.5);
	\draw[dualEdge] (1.5,1.5) -- (2.5,1.5);
	\draw[dualEdge] (1.5,1.5) -- (1.5,2.5);
	\draw[dualEdge] (-0.5,-0.5) -- (-0.5,0.5);
	\draw[dualEdge] (-0.5,-0.5) -- (0.5,-0.5);
	\draw[dualEdge] (-0.5,0.5) -- (0.5,0.5);
	\draw[dualEdge] (-0.5,0.5) -- (-0.5,1.5);
	\draw[dualEdge] (0.5,0.5) -- (1.5,0.5);
	\draw[dualEdge] (0.5,0.5) -- (0.5,1.5);
	\draw[dualEdge] (0.5,-0.5) -- (0.5,0.5);
	\draw[dualEdge] (0.5,-0.5) -- (1.5,-0.5);
	\draw[dualEdge] (-1.5,2.5) -- (-0.5,2.5);
	\draw[dualEdge] (1.5,0.5) -- (1.5,1.5);
	\draw[dualEdge] (1.5,0.5) -- (2.5,0.5);
	\draw[dualEdge] (1.5,-0.5) -- (1.5,0.5);
	\draw[dualEdge] (0.5,1.5) -- (0.5,2.5);
	\draw[dualEdge] (0.5,1.5) -- (1.5,1.5);
	\draw[dualEdge] (-0.5,1.5) -- (0.5,1.5);
	\draw[dualEdge] (-0.5,1.5) -- (-0.5,2.5);
	\draw[dualEdge] (3.5,1.5) -- (3.5,2.5);
	\draw[dualEdge] (2.5,0.5) -- (3.5,0.5);
	\draw[dualEdge] (2.5,0.5) -- (2.5,1.5);
	\draw[dualEdge] (2.5,-0.5) -- (3.5,-0.5);
	\draw[dualEdge] (2.5,-0.5) -- (2.5,0.5);
	\draw[dualEdge] (-0.5,2.5) -- (0.5,2.5);
	\draw[dualEdge] (-0.5,2.5) -- (-0.5,3.5);
	\draw[dualEdge] (-1.5,1.5) -- (-1.5,2.5);
	\draw[dualEdge] (-1.5,1.5) -- (-0.5,1.5);
	\draw[contour] (2.5,1.5) -- (3.5,1.5);
	\draw[contour] (-0.5,1.5) -- (0.5,1.5);
	\draw[contour] (0.5,2.5) -- (1.5,2.5);
	\draw[contour] (2.5,2.5) -- (3.5,2.5);
	\draw[contour] (1.5,0.5) -- (2.5,0.5);
	\draw[contour] (0.5,2.5) -- (0.5,3.5);
	\draw[contour] (2.5,1.5) -- (2.5,2.5);
	\draw[contour] (-0.5,0.5) -- (0.5,0.5);
	\draw[contour] (-0.5,2.5) -- (0.5,2.5);
	\draw[contour] (3.5,1.5) -- (3.5,2.5);
	\draw[contour] (0.5,3.5) -- (1.5,3.5);
	\draw[contour] (1.5,2.5) -- (1.5,3.5);
	\draw[contour] (0.5,1.5) -- (1.5,1.5);
	\draw[contour] (2.5,-0.5) -- (2.5,0.5);
	\draw[contour] (1.5,0.5) -- (1.5,1.5);
	\draw[contour] (-0.5,1.5) -- (-0.5,2.5);
	\draw[contour] (0.5,1.5) -- (0.5,2.5);
	\draw[contour] (2.5,-0.5) -- (3.5,-0.5);
	\draw[contour] (-0.5,0.5) -- (-0.5,1.5);
	\draw[contour] (0.5,0.5) -- (0.5,1.5);
	\draw[contour] (3.5,-0.5) -- (3.5,0.5);
	\draw[contour] (-1.5,1.5) -- (-0.5,1.5);
	\draw (1.5,3.5) node[face] {};
	\draw (3.5,0.5) node[face] {};
	\draw (3.5,-0.5) node[face] {};
	\draw (0.5,2.5) node[face] {};
	\draw (2.5,1.5) node[face] {};
	\draw (0.5,3.5) node[face] {};
	\draw (-0.5,3.5) node[face] {};
	\draw (1.5,2.5) node[face] {};
	\draw (2.5,2.5) node[face] {};
	\draw (1.5,1.5) node[face] {};
	\draw (3.5,2.5) node[face] {};
	\draw (-0.5,-0.5) node[face] {};
	\draw (-0.5,0.5) node[face] {};
	\draw (0.5,0.5) node[face] {};
	\draw (0.5,-0.5) node[face] {};
	\draw (-1.5,2.5) node[face] {};
	\draw (1.5,0.5) node[face] {};
	\draw (1.5,-0.5) node[face] {};
	\draw (0.5,1.5) node[face] {};
	\draw (-0.5,1.5) node[face] {};
	\draw (3.5,1.5) node[face] {};
	\draw (2.5,0.5) node[face] {};
	\draw (2.5,-0.5) node[face] {};
	\draw (-0.5,2.5) node[face] {};
	\draw (-1.5,1.5) node[face] {};
	\draw (-1.5,1.5) node[dobFace] {};
	\draw (3.5,0.5) node[dobFace] {};
\end{tikzpicture}
\caption{A contour configuration $P \in \cP_\Omega^{\xi}$ with the same $\Omega$ and $\xi$ as in the preceding figures.
Here, only the marked faces $\bface$ and $\eface$ are highlighted as ``\raisebox{-1mm}{\tikzvertex{face,fill=green!30}}''. 
While the marked dual edges $\bedge$ and $\eedge$ do not play any role in $\cP_\Omega^\xi$ when lying on the boundary, 
they are important for the relation with Grassmann integrals (see \cref{sec:Grassman}) and for the generalized boundary conditions used to define martingale observables in \cref{sec:interface}. 
}
\label{fig:P_example}
\end{figure}

We say that a pair $\xi = (\bface,\bedge) \in \cVcluster\Omega \times \cVcluster\Omega$ is an \emph{admissible boundary condition} if there exist $\Ombc \supset \Omega$ and $\Pbc \subset \cEdual{\Ombc\setminus \Omega}$ with $|\Ombc| < \infty$ such that $P \cup \Pbc \in \cP_{\Ombc}$ for all $P \in \cP_\Omega^\xi$.
In particular this requires that $\bhalfedge,\ehalfedge$ lie on the boundary of $\Omega$ as shown in \cref{fig:Dobrushin}.
\begin{figure}[t]
\centering
\begin{subfigure}[t]{0.45\textwidth}
\tikzpicturedependsonfile{dobrushin.tikz}
\begin{tikzpicture}
	\draw (0,0) -- (0,1);
	\draw (0,0) -- (1,0);
	\draw (0,1) -- (0,2);
	\draw (0,1) -- (1,1);
	\draw (0,2) -- (0,3);
	\draw (0,2) -- (1,2);
	\draw (0,3) -- (1,3);
	\draw (1,0) -- (1,1);
	\draw (1,1) -- (1,2);
	\draw (1,1) -- (2,1);
	\draw (1,2) -- (1,3);
	\draw (2,1) -- (3,1);
	\draw (3,0) -- (3,1);
	\draw (3,1) -- (3,2);
	\draw (-1,2) -- (0,2);
	\draw (0,0) -- (-1,0);
	\draw (0,0) -- (0,-1);
	\draw (0,1) -- (-1,1);
	\draw (0,3) -- (-1,3);
	\draw (0,3) -- (0,4);
	\draw (1,0) -- (2,0);
	\draw (1,0) -- (1,-1);
	\draw (1,2) -- (2,2);
	\draw (1,3) -- (1,4);
	\draw (1,3) -- (2,3);
	\draw (2,1) -- (2,2);
	\draw (2,1) -- (2,0);
	\draw (3,0) -- (3,-1);
	\draw (3,0) -- (2,0);
	\draw (3,0) -- (4,0);
	\draw (3,1) -- (4,1);
	\draw (3,2) -- (3,3);
	\draw (3,2) -- (2,2);
	\draw (3,2) -- (4,2);
	\draw (-1,2) -- (-1,1);
	\draw (-1,2) -- (-1,3);
	\draw (-1,2) -- (-2,2);
	\draw[dualEdge] (3.5,0.5) -- (3.5,1.5);
	\draw[dualEdge] (3.5,-0.5) -- (3.5,0.5);
	\draw[dualEdge] (0.5,2.5) -- (0.5,3.5);
	\draw[dualEdge] (0.5,2.5) -- (1.5,2.5);
	\draw[dualEdge] (2.5,1.5) -- (2.5,2.5);
	\draw[dualEdge] (2.5,1.5) -- (3.5,1.5);
	\draw[dualEdge] (0.5,3.5) -- (1.5,3.5);
	\draw[dualEdge] (-0.5,3.5) -- (0.5,3.5);
	\draw[dualEdge] (1.5,2.5) -- (1.5,3.5);
	\draw[dualEdge] (2.5,2.5) -- (3.5,2.5);
	\draw[dualEdge] (1.5,1.5) -- (2.5,1.5);
	\draw[dualEdge] (1.5,1.5) -- (1.5,2.5);
	\draw[dualEdge] (-0.5,-0.5) -- (-0.5,0.5);
	\draw[dualEdge] (-0.5,-0.5) -- (0.5,-0.5);
	\draw[dualEdge] (-0.5,0.5) -- (0.5,0.5);
	\draw[dualEdge] (-0.5,0.5) -- (-0.5,1.5);
	\draw[dualEdge] (0.5,0.5) -- (1.5,0.5);
	\draw[dualEdge] (0.5,0.5) -- (0.5,1.5);
	\draw[dualEdge] (0.5,-0.5) -- (0.5,0.5);
	\draw[dualEdge] (0.5,-0.5) -- (1.5,-0.5);
	\draw[dualEdge] (-1.5,2.5) -- (-0.5,2.5);
	\draw[dualEdge] (1.5,0.5) -- (1.5,1.5);
	\draw[dualEdge] (1.5,0.5) -- (2.5,0.5);
	\draw[dualEdge] (1.5,-0.5) -- (1.5,0.5);
	\draw[dualEdge] (0.5,1.5) -- (0.5,2.5);
	\draw[dualEdge] (0.5,1.5) -- (1.5,1.5);
	\draw[dualEdge] (-0.5,1.5) -- (0.5,1.5);
	\draw[dualEdge] (-0.5,1.5) -- (-0.5,2.5);
	\draw[dualEdge] (3.5,1.5) -- (3.5,2.5);
	\draw[dualEdge] (2.5,0.5) -- (3.5,0.5);
	\draw[dualEdge] (2.5,0.5) -- (2.5,1.5);
	\draw[dualEdge] (2.5,-0.5) -- (3.5,-0.5);
	\draw[dualEdge] (2.5,-0.5) -- (2.5,0.5);
	\draw[dualEdge] (-0.5,2.5) -- (0.5,2.5);
	\draw[dualEdge] (-0.5,2.5) -- (-0.5,3.5);
	\draw[dualEdge] (-1.5,1.5) -- (-1.5,2.5);
	\draw[dualEdge] (-1.5,1.5) -- (-0.5,1.5);
	\draw (1.5,3.5) node[face] {};
	\draw (3.5,0.5) node[face] {};
	\draw (3.5,-0.5) node[face] {};
	\draw (0.5,2.5) node[face] {};
	\draw (2.5,1.5) node[face] {};
	\draw (0.5,3.5) node[face] {};
	\draw (-0.5,3.5) node[face] {};
	\draw (1.5,2.5) node[face] {};
	\draw (2.5,2.5) node[face] {};
	\draw (1.5,1.5) node[face] {};
	\draw (3.5,2.5) node[face] {};
	\draw (-0.5,-0.5) node[face] {};
	\draw (-0.5,0.5) node[face] {};
	\draw (0.5,0.5) node[face] {};
	\draw (0.5,-0.5) node[face] {};
	\draw (-1.5,2.5) node[face] {};
	\draw (1.5,0.5) node[face] {};
	\draw (1.5,-0.5) node[face] {};
	\draw (0.5,1.5) node[face] {};
	\draw (-0.5,1.5) node[face] {};
	\draw (3.5,1.5) node[face] {};
	\draw (2.5,0.5) node[face] {};
	\draw (2.5,-0.5) node[face] {};
	\draw (-0.5,2.5) node[face] {};
	\draw (-1.5,1.5) node[face] {};
	\draw (0,0) node[vertex] {};
	\draw (0,1) node[vertex] {};
	\draw (0,2) node[vertex] {};
	\draw (0,3) node[vertex] {};
	\draw (1,0) node[vertex] {};
	\draw (1,1) node[vertex] {};
	\draw (1,2) node[vertex] {};
	\draw (1,3) node[vertex] {};
	\draw (2,1) node[vertex] {};
	\draw (3,0) node[vertex] {};
	\draw (3,1) node[vertex] {};
	\draw (3,2) node[vertex] {};
	\draw (-1,2) node[vertex] {};
	\draw (-1,1) node[edgeVertex] {};
	\draw (4,1) node[edgeVertex] {};
	\draw (-1,0) node[edgeVertex] {};
	\draw (3,3) node[edgeVertex] {};
	\draw (-1,3) node[edgeVertex] {};
	\draw (2,2) node[edgeVertex] {};
	\draw (-2,2) node[edgeVertex] {};
	\draw (1,4) node[edgeVertex] {};
	\draw (3,-1) node[edgeVertex] {};
	\draw (2,0) node[edgeVertex] {};
	\draw (2,3) node[edgeVertex] {};
	\draw (0,4) node[edgeVertex] {};
	\draw (4,2) node[edgeVertex] {};
	\draw (0,-1) node[edgeVertex] {};
	\draw (1,-1) node[edgeVertex] {};
	\draw (4,0) node[edgeVertex] {};
	\draw[dobEdge] (-1.5,1.5) -- (-1.5,0.7);
	\draw (-1.5,1.5) node[dobFace] {};
	\draw[dobEdge] (3.5,0.5) -- (4.3,0.5);
	\draw (3.5,0.5) node[dobFace] {};
\end{tikzpicture}
\caption{An admissible boundary condition $\xi = (\bhalfedge,\ehalfedge)$ for the above graph, with the faces $\bface$ and $\eface$ shown as ``\raisebox{-1mm}{\tikzvertex{dobFace}}'' 
and the external dual edges $\bedge$ and $\eedge$ in green. 
}
\end{subfigure}
~
\begin{subfigure}[t]{0.45\textwidth}
\centering
\tikzpicturedependsonfile{Kasteleyn_dobrushin.tikz}
\begin{tikzpicture}
	\draw (-1,1) node[vertex,draw=white,fill=white] {};
	\draw (4,1) node[vertex,draw=white,fill=white] {};
	\draw (-1,0) node[vertex,draw=white,fill=white] {};
	\draw (3,3) node[vertex,draw=white,fill=white] {};
	\draw (-1,3) node[vertex,draw=white,fill=white] {};
	\draw (2,2) node[vertex,draw=white,fill=white] {};
	\draw (-2,2) node[vertex,draw=white,fill=white] {};
	\draw (1,4) node[vertex,draw=white,fill=white] {};
	\draw (3,-1) node[vertex,draw=white,fill=white] {};
	\draw (2,0) node[vertex,draw=white,fill=white] {};
	\draw (2,3) node[vertex,draw=white,fill=white] {};
	\draw (0,4) node[vertex,draw=white,fill=white] {};
	\draw (4,2) node[vertex,draw=white,fill=white] {};
	\draw (0,-1) node[vertex,draw=white,fill=white] {};
	\draw (1,-1) node[vertex,draw=white,fill=white] {};
	\draw (4,0) node[vertex,draw=white,fill=white] {};
	\draw[clusterLongEdge] (3.5,0.5) -- (3.5,1.5);
	\draw[clusterLongEdge] (3.5,-0.5) -- (3.5,0.5);
	\draw[clusterLongEdge] (0.5,2.5) -- (0.5,3.5);
	\draw[clusterLongEdge] (0.5,2.5) -- (1.5,2.5);
	\draw[clusterLongEdge] (2.5,1.5) -- (2.5,2.5);
	\draw[clusterLongEdge] (2.5,1.5) -- (3.5,1.5);
	\draw[clusterLongEdge] (0.5,3.5) -- (1.5,3.5);
	\draw[clusterLongEdge] (-0.5,3.5) -- (0.5,3.5);
	\draw[clusterLongEdge] (1.5,2.5) -- (1.5,3.5);
	\draw[clusterLongEdge] (2.5,2.5) -- (3.5,2.5);
	\draw[clusterLongEdge] (1.5,1.5) -- (2.5,1.5);
	\draw[clusterLongEdge] (1.5,1.5) -- (1.5,2.5);
	\draw[clusterLongEdge] (-0.5,-0.5) -- (-0.5,0.5);
	\draw[clusterLongEdge] (-0.5,-0.5) -- (0.5,-0.5);
	\draw[clusterLongEdge] (-0.5,0.5) -- (0.5,0.5);
	\draw[clusterLongEdge] (-0.5,0.5) -- (-0.5,1.5);
	\draw[clusterLongEdge] (0.5,0.5) -- (1.5,0.5);
	\draw[clusterLongEdge] (0.5,0.5) -- (0.5,1.5);
	\draw[clusterLongEdge] (0.5,-0.5) -- (0.5,0.5);
	\draw[clusterLongEdge] (0.5,-0.5) -- (1.5,-0.5);
	\draw[clusterLongEdge] (-1.5,2.5) -- (-0.5,2.5);
	\draw[clusterLongEdge] (1.5,0.5) -- (1.5,1.5);
	\draw[clusterLongEdge] (1.5,0.5) -- (2.5,0.5);
	\draw[clusterLongEdge] (1.5,-0.5) -- (1.5,0.5);
	\draw[clusterLongEdge] (0.5,1.5) -- (0.5,2.5);
	\draw[clusterLongEdge] (0.5,1.5) -- (1.5,1.5);
	\draw[clusterLongEdge] (-0.5,1.5) -- (0.5,1.5);
	\draw[clusterLongEdge] (-0.5,1.5) -- (-0.5,2.5);
	\draw[clusterLongEdge] (3.5,1.5) -- (3.5,2.5);
	\draw[clusterLongEdge] (2.5,0.5) -- (3.5,0.5);
	\draw[clusterLongEdge] (2.5,0.5) -- (2.5,1.5);
	\draw[clusterLongEdge] (2.5,-0.5) -- (3.5,-0.5);
	\draw[clusterLongEdge] (2.5,-0.5) -- (2.5,0.5);
	\draw[clusterLongEdge] (-0.5,2.5) -- (0.5,2.5);
	\draw[clusterLongEdge] (-0.5,2.5) -- (-0.5,3.5);
	\draw[clusterLongEdge] (-1.5,1.5) -- (-1.5,2.5);
	\draw[clusterLongEdge] (-1.5,1.5) -- (-0.5,1.5);
	\draw (1.5,3.5) node[clusterFace] {};
	\draw (3.5,0.5) node[clusterFace] {};
	\draw (3.5,-0.5) node[clusterFace] {};
	\draw (0.5,2.5) node[clusterFace] {};
	\draw (2.5,1.5) node[clusterFace] {};
	\draw (0.5,3.5) node[clusterFace] {};
	\draw (-0.5,3.5) node[clusterFace] {};
	\draw (1.5,2.5) node[clusterFace] {};
	\draw (2.5,2.5) node[clusterFace] {};
	\draw (1.5,1.5) node[clusterFace] {};
	\draw (3.5,2.5) node[clusterFace] {};
	\draw (-0.5,-0.5) node[clusterFace] {};
	\draw (-0.5,0.5) node[clusterFace] {};
	\draw (0.5,0.5) node[clusterFace] {};
	\draw (0.5,-0.5) node[clusterFace] {};
	\draw (-1.5,2.5) node[clusterFace] {};
	\draw (1.5,0.5) node[clusterFace] {};
	\draw (1.5,-0.5) node[clusterFace] {};
	\draw (0.5,1.5) node[clusterFace] {};
	\draw (-0.5,1.5) node[clusterFace] {};
	\draw (3.5,1.5) node[clusterFace] {};
	\draw (2.5,0.5) node[clusterFace] {};
	\draw (2.5,-0.5) node[clusterFace] {};
	\draw (-0.5,2.5) node[clusterFace] {};
	\draw (-1.5,1.5) node[clusterFace] {};
	\draw (-1.5,1.5) node[clusterFaceDobIn] {};
	\draw (3.5,0.5) node[clusterFaceDobOut] {};
\end{tikzpicture}
\caption{The corresponding cluster graph, with the half-edges $\bhalfedge,\ehalfedge$ in green.}
\end{subfigure}
\caption{Admissible boundary conditions.}
\label{fig:Dobrushin}
\end{figure}
Then, $P \cup P_o$ is associated with a unique spin configuration on $\Omega_o$ 
and in particular, there are two components of the boundary $\partial \cVprimal\Omega$ depending only on $\xi$ 
which are always assigned opposite values, depending only on $P_o$, as shown in \cref{fig:Ising_Dobrushin}, 
what are known as \emph{Dobrushin boundary conditions}. 
In terms of the contours, the associated partition function is then given by
\begin{align}
Z_{\Omega,\ul x}^{\xi}
=
\sum_{P \in \cP^{\xi}_\Omega} \prod_{\edual \in P} x_{\edual} , \qquad \xi = (\bhalfedge,\ehalfedge) ,
\label{eq:Z_Dob_lt}
\end{align}
and the associated probability measure is\footnote{Note that these objects do not actually depend on $\Omega_0$ and $P_o$ except via $\xi$.  In fact, the definitions make perfect sense even when $\xi$ is not admissible (although the relationship to the Gibbs measure of the Ising model breaks down) and we use this generalization in an important way in the following section.} 
\begin{align}
\label{eq:contour_probability_measure}
\bP_{\Omega,\ul x}^{\xi}[P] := \frac{1}{Z^{\xi}_{\Omega,\ul x}} \, \prod_{\edual \in P} x_{\edual} , \qquad P \in \cP^{\xi}_\Omega .
\end{align}

\begin{remark}
All of this can easily be generalized to $\xi=(\halfedge_1,\ldots,\halfedge_{2N})$ with $N \geq 2$, 
corresponding to alternating boundary conditions as in~\cite{Izyurov:Critical_Ising_interfaces_in_multiply_connected_domains}, and with slightly more attention this is true for what follows as well 
(but we will mostly stick to the case with two marked points to avoid notational complications).
\end{remark}

\begin{figure}[h!]
\centering
\tikzpicturedependsonfile{Ising_dobrushin.tikz}
\begin{tikzpicture}
	\draw (0,0) -- (0,1);
	\draw (0,0) -- (1,0);
	\draw (0,1) -- (0,2);
	\draw (0,1) -- (1,1);
	\draw (0,2) -- (0,3);
	\draw (0,2) -- (1,2);
	\draw (0,3) -- (1,3);
	\draw (1,0) -- (1,1);
	\draw (1,1) -- (1,2);
	\draw (1,1) -- (2,1);
	\draw (1,2) -- (1,3);
	\draw (2,1) -- (3,1);
	\draw (3,0) -- (3,1);
	\draw (3,1) -- (3,2);
	\draw (-1,2) -- (0,2);
	\draw (0,0) -- (-1,0);
	\draw (0,0) -- (0,-1);
	\draw (0,1) -- (-1,1);
	\draw (0,3) -- (-1,3);
	\draw (0,3) -- (0,4);
	\draw (1,0) -- (2,0);
	\draw (1,0) -- (1,-1);
	\draw (1,2) -- (2,2);
	\draw (1,3) -- (1,4);
	\draw (1,3) -- (2,3);
	\draw (2,1) -- (2,2);
	\draw (2,1) -- (2,0);
	\draw (3,0) -- (3,-1);
	\draw (3,0) -- (2,0);
	\draw (3,0) -- (4,0);
	\draw (3,1) -- (4,1);
	\draw (3,2) -- (3,3);
	\draw (3,2) -- (2,2);
	\draw (3,2) -- (4,2);
	\draw (-1,2) -- (-1,1);
	\draw (-1,2) -- (-1,3);
	\draw (-1,2) -- (-2,2);
	\draw[dualEdge] (3.5,0.5) -- (3.5,1.5);
	\draw[dualEdge] (3.5,-0.5) -- (3.5,0.5);
	\draw[dualEdge] (0.5,2.5) -- (0.5,3.5);
	\draw[dualEdge] (0.5,2.5) -- (1.5,2.5);
	\draw[dualEdge] (2.5,1.5) -- (2.5,2.5);
	\draw[dualEdge] (2.5,1.5) -- (3.5,1.5);
	\draw[dualEdge] (0.5,3.5) -- (1.5,3.5);
	\draw[dualEdge] (-0.5,3.5) -- (0.5,3.5);
	\draw[dualEdge] (1.5,2.5) -- (1.5,3.5);
	\draw[dualEdge] (2.5,2.5) -- (3.5,2.5);
	\draw[dualEdge] (1.5,1.5) -- (2.5,1.5);
	\draw[dualEdge] (1.5,1.5) -- (1.5,2.5);
	\draw[dualEdge] (-0.5,-0.5) -- (-0.5,0.5);
	\draw[dualEdge] (-0.5,-0.5) -- (0.5,-0.5);
	\draw[dualEdge] (-0.5,0.5) -- (0.5,0.5);
	\draw[dualEdge] (-0.5,0.5) -- (-0.5,1.5);
	\draw[dualEdge] (0.5,0.5) -- (1.5,0.5);
	\draw[dualEdge] (0.5,0.5) -- (0.5,1.5);
	\draw[dualEdge] (0.5,-0.5) -- (0.5,0.5);
	\draw[dualEdge] (0.5,-0.5) -- (1.5,-0.5);
	\draw[dualEdge] (-1.5,2.5) -- (-0.5,2.5);
	\draw[dualEdge] (1.5,0.5) -- (1.5,1.5);
	\draw[dualEdge] (1.5,0.5) -- (2.5,0.5);
	\draw[dualEdge] (1.5,-0.5) -- (1.5,0.5);
	\draw[dualEdge] (0.5,1.5) -- (0.5,2.5);
	\draw[dualEdge] (0.5,1.5) -- (1.5,1.5);
	\draw[dualEdge] (-0.5,1.5) -- (0.5,1.5);
	\draw[dualEdge] (-0.5,1.5) -- (-0.5,2.5);
	\draw[dualEdge] (3.5,1.5) -- (3.5,2.5);
	\draw[dualEdge] (2.5,0.5) -- (3.5,0.5);
	\draw[dualEdge] (2.5,0.5) -- (2.5,1.5);
	\draw[dualEdge] (2.5,-0.5) -- (3.5,-0.5);
	\draw[dualEdge] (2.5,-0.5) -- (2.5,0.5);
	\draw[dualEdge] (-0.5,2.5) -- (0.5,2.5);
	\draw[dualEdge] (-0.5,2.5) -- (-0.5,3.5);
	\draw[dualEdge] (-1.5,1.5) -- (-1.5,2.5);
	\draw[dualEdge] (-1.5,1.5) -- (-0.5,1.5);
	\draw (1.5,3.5) node[face] {};
	\draw (3.5,0.5) node[face] {};
	\draw (3.5,-0.5) node[face] {};
	\draw (0.5,2.5) node[face] {};
	\draw (2.5,1.5) node[face] {};
	\draw (0.5,3.5) node[face] {};
	\draw (-0.5,3.5) node[face] {};
	\draw (1.5,2.5) node[face] {};
	\draw (2.5,2.5) node[face] {};
	\draw (1.5,1.5) node[face] {};
	\draw (3.5,2.5) node[face] {};
	\draw (-0.5,-0.5) node[face] {};
	\draw (-0.5,0.5) node[face] {};
	\draw (0.5,0.5) node[face] {};
	\draw (0.5,-0.5) node[face] {};
	\draw (-1.5,2.5) node[face] {};
	\draw (1.5,0.5) node[face] {};
	\draw (1.5,-0.5) node[face] {};
	\draw (0.5,1.5) node[face] {};
	\draw (-0.5,1.5) node[face] {};
	\draw (3.5,1.5) node[face] {};
	\draw (2.5,0.5) node[face] {};
	\draw (2.5,-0.5) node[face] {};
	\draw (-0.5,2.5) node[face] {};
	\draw (-1.5,1.5) node[face] {};
	\draw (0,0) node[vertex] {};
	\draw (0,1) node[vertex] {};
	\draw (0,2) node[vertex] {};
	\draw (0,3) node[vertex] {};
	\draw (1,0) node[vertex] {};
	\draw (1,1) node[vertex] {};
	\draw (1,2) node[vertex] {};
	\draw (1,3) node[vertex] {};
	\draw (2,1) node[vertex] {};
	\draw (3,0) node[vertex] {};
	\draw (3,1) node[vertex] {};
	\draw (3,2) node[vertex] {};
	\draw (-1,2) node[vertex] {};
	\draw (-1,1) node[edgeVertexSpinMinus] {};
	\draw (4,1) node[edgeVertexSpinPlus] {};
	\draw (-1,0) node[edgeVertexSpinMinus] {};
	\draw (3,3) node[edgeVertexSpinPlus] {};
	\draw (-1,3) node[edgeVertexSpinPlus] {};
	\draw (2,2) node[edgeVertexSpinPlus] {};
	\draw (-2,2) node[edgeVertexSpinPlus] {};
	\draw (1,4) node[edgeVertexSpinPlus] {};
	\draw (3,-1) node[edgeVertexSpinMinus] {};
	\draw (2,0) node[edgeVertexSpinMinus] {};
	\draw (2,3) node[edgeVertexSpinPlus] {};
	\draw (0,4) node[edgeVertexSpinPlus] {};
	\draw (4,2) node[edgeVertexSpinPlus] {};
	\draw (0,-1) node[edgeVertexSpinMinus] {};
	\draw (1,-1) node[edgeVertexSpinMinus] {};
	\draw (4,0) node[edgeVertexSpinMinus] {};
	\draw[dobEdge] (-1.5,1.5) -- (-1.5,0.7);
	\draw (-1.5,1.5) node[dobFace] {};
	\draw[dobEdge] (3.5,0.5) -- (4.3,0.5);
	\draw (3.5,0.5) node[dobFace] {};
\end{tikzpicture}
\caption{An example of a choice of the faces $\bface$ and $\eface$ (shown as ``\raisebox{-1mm}{\tikzvertex{dobFace}}'') 
and the external dual edges $\bedge$ and $\eedge$
 specifying a choice of Dobrushin boundary conditions. 
 The boundary is divided into two components, drawn as ``\tikzvertex{edgeVertexSpinPlus}'' and ``\tikzvertex{edgeVertexSpinMinus}''. 
}
\label{fig:Ising_Dobrushin}
\end{figure}

\subsection{Exploration interface and martingale observable}
\label{sec:interface}

The following discussion makes no direct connection to the spin Ising model; it is valid in general for any probability measure on the contour space $\smash{\cP^\xi_\Omega}$ --- including the representation of the more general model introduced in \cref{sec:interact}, which does not have the same form as \cref{eq:contour_probability_measure} 
(see in particular \cref{thm:main_thm}). 

\medskip

For $\xi = (\bhalfedge,\ehalfedge)$
we can define the \emph{Peierls interface}, or \emph{exploration path}  
(with Dobrushin boundary conditions) as a path from 
a given boundary vertex $\bhalfedge = (\bface,\bedge)$ to another vertex $\ehalfedge = (\eface,\eedge)$ consisting of edges which are included in $P$, with no edge traversed more than once.  
In terms of the spins, when $\xi$ is an admissible boundary condition, 
one can think of the interface as separating spins of different values 
(however we continue to make definitions in the general case not admitting such an interpretation).
For $P \in \smash{\cP^\xi_\Omega}$ with $\xi = (\bhalfedge,\ehalfedge)$, 
the parity constraint guarantees that such a path exists, but not that it would be unique (due to an ambiguity for faces with degree $\degree_f(P) = 4$). 
We shall fix a unique choice by the following procedure, illustrated in \cref{fig:contour}.

\begin{df} \label{def:interface}
We iteratively construct sequences $(\gamma_n)_{n \in \bZnn} = (\gamma_n(P))_{n \in \bZnn}$ of dual edges 
and $(\vec \gamma_n)_{n \in \bZnn} = (\vec \gamma_n(P))_{n \in \bZnn}$ of half-edges,  
starting at the initial edge 
$\gamma_0 := \bedge$ and half-edge $\vec \gamma_0 := \bhalfedge$,  and iterating as follows (see \cref{fig:contour}). 
\begin{enumerate}
\item If $\gamma_n = \eedge$, then stop.

\item If only one dual edge $\edual \in P \cup \{ \eedge \}$ which does not already appear in $\{\gamma_0,\ldots,\gamma_n\}$ is incident to $f(\vec \gamma_n)$, then set $\gamma_{n+1} := \edual$ 
and set $\vec \gamma_{n+1} := (f',e)$, where $f'$ is the other face incident to $\edual$.

\item If neither of the above holds, then the parity constraint requires that all of the other edges incident to $f(\vec \gamma_n)$ belong to $P \cup \{ \eedge \}$.  
In this case, set $\gamma_{n+1}$ to be the edge such that $\{ \gamma_n, \gamma_{n+1} \} = \{ f(\vec \gamma_n) \pm 1/2, f(\vec \gamma_n) \pm \ii/2 \}$
and set $\vec \gamma_{n+1}$, where $f'$ is the other face incident to~$\gamma_{n+1}$.
(This corresponds to resolving the ambiguity by pairing edges in 
the interface in the North-East and South-West directions\footnote{This is not the only way to define an interface which gives the same scaling limit.  
However some choices, for example the construction used in~\cite{Chelkak-Smirnov:Universality_in_2D_Ising_and_conformal_invariance_of_fermionic_observables, Hongler-Kytola:Ising_interfaces_and_free_boundary_conditions}, are not reversible.\label{foot:reversible}}, and we refer to it later as the NE/SW rule and the preferred pairs of edges as NE/SW pairs). 
\end{enumerate}
\end{df}

\begin{figure}[h!]
\centering
\tikzpicturedependsonfile{dobContour.tikz}
\begin{tikzpicture}
	\draw[dualEdge] (3.5,0.5) -- (3.5,1.5);
	\draw[dualEdge] (3.5,-0.5) -- (3.5,0.5);
	\draw[dualEdge] (0.5,2.5) -- (0.5,3.5);
	\draw[dualEdge] (0.5,2.5) -- (1.5,2.5);
	\draw[dualEdge] (2.5,1.5) -- (2.5,2.5);
	\draw[dualEdge] (2.5,1.5) -- (3.5,1.5);
	\draw[dualEdge] (0.5,3.5) -- (1.5,3.5);
	\draw[dualEdge] (-0.5,3.5) -- (0.5,3.5);
	\draw[dualEdge] (1.5,2.5) -- (1.5,3.5);
	\draw[dualEdge] (2.5,2.5) -- (3.5,2.5);
	\draw[dualEdge] (1.5,1.5) -- (2.5,1.5);
	\draw[dualEdge] (1.5,1.5) -- (1.5,2.5);
	\draw[dualEdge] (-0.5,-0.5) -- (-0.5,0.5);
	\draw[dualEdge] (-0.5,-0.5) -- (0.5,-0.5);
	\draw[dualEdge] (-0.5,0.5) -- (0.5,0.5);
	\draw[dualEdge] (-0.5,0.5) -- (-0.5,1.5);
	\draw[dualEdge] (0.5,0.5) -- (1.5,0.5);
	\draw[dualEdge] (0.5,0.5) -- (0.5,1.5);
	\draw[dualEdge] (0.5,-0.5) -- (0.5,0.5);
	\draw[dualEdge] (0.5,-0.5) -- (1.5,-0.5);
	\draw[dualEdge] (-1.5,2.5) -- (-0.5,2.5);
	\draw[dualEdge] (1.5,0.5) -- (1.5,1.5);
	\draw[dualEdge] (1.5,0.5) -- (2.5,0.5);
	\draw[dualEdge] (1.5,-0.5) -- (1.5,0.5);
	\draw[dualEdge] (0.5,1.5) -- (0.5,2.5);
	\draw[dualEdge] (0.5,1.5) -- (1.5,1.5);
	\draw[dualEdge] (-0.5,1.5) -- (0.5,1.5);
	\draw[dualEdge] (-0.5,1.5) -- (-0.5,2.5);
	\draw[dualEdge] (3.5,1.5) -- (3.5,2.5);
	\draw[dualEdge] (2.5,0.5) -- (3.5,0.5);
	\draw[dualEdge] (2.5,0.5) -- (2.5,1.5);
	\draw[dualEdge] (2.5,-0.5) -- (3.5,-0.5);
	\draw[dualEdge] (2.5,-0.5) -- (2.5,0.5);
	\draw[dualEdge] (-0.5,2.5) -- (0.5,2.5);
	\draw[dualEdge] (-0.5,2.5) -- (-0.5,3.5);
	\draw[dualEdge] (-1.5,1.5) -- (-1.5,2.5);
	\draw[dualEdge] (-1.5,1.5) -- (-0.5,1.5);
	\draw[contour] (2.5,1.5) -- (3.5,1.5);
	\draw[contour] (-0.5,1.5) -- (0.5,1.5);
	\draw[contour] (0.5,2.5) -- (1.5,2.5);
	\draw[contour] (2.5,2.5) -- (3.5,2.5);
	\draw[contour] (1.5,0.5) -- (2.5,0.5);
	\draw[contour] (0.5,2.5) -- (0.5,3.5);
	\draw[contour] (2.5,1.5) -- (2.5,2.5);
	\draw[contour] (-0.5,0.5) -- (0.5,0.5);
	\draw[contour] (-0.5,2.5) -- (0.5,2.5);
	\draw[contour] (3.5,1.5) -- (3.5,2.5);
	\draw[contour] (0.5,3.5) -- (1.5,3.5);
	\draw[contour] (1.5,2.5) -- (1.5,3.5);
	\draw[contour] (0.5,1.5) -- (1.5,1.5);
	\draw[contour] (2.5,-0.5) -- (2.5,0.5);
	\draw[contour] (1.5,0.5) -- (1.5,1.5);
	\draw[contour] (-0.5,1.5) -- (-0.5,2.5);
	\draw[contour] (0.5,1.5) -- (0.5,2.5);
	\draw[contour] (2.5,-0.5) -- (3.5,-0.5);
	\draw[contour] (-0.5,0.5) -- (-0.5,1.5);
	\draw[contour] (0.5,0.5) -- (0.5,1.5);
	\draw[contour] (3.5,-0.5) -- (3.5,0.5);
	\draw[contour] (-1.5,1.5) -- (-0.5,1.5);
	\draw[peierls] (-1.5,0.5) --(-1.5,1.5) node [midway,label={[label distance = 1pt]left: $\scriptstyle0$}] {} --(-0.5,1.5) node [midway,label={[label distance = 1pt]below: $\scriptstyle1$}] {} --(-0.5,0.5) node [midway,label={[label distance = 1pt]left: $\scriptstyle2$}] {} --(0.5,0.5) node [midway,label={[label distance = 1pt]below: $\scriptstyle3$}] {} --(0.5,1.5) node [midway,label={[label distance = 1pt]left: $\scriptstyle4$}] {} --(-0.5,1.5) node [midway,label={[label distance = 1pt]below: $\scriptstyle5$}] {} --(-0.5,2.5) node [midway,label={[label distance = 1pt]left: $\scriptstyle6$}] {} --(0.5,2.5) node [midway,label={[label distance = 1pt]below: $\scriptstyle7$}] {} --(0.5,1.5) node [midway,label={[label distance = 1pt]left: $\scriptstyle8$}] {} --(1.5,1.5) node [midway,label={[label distance = 1pt]below: $\scriptstyle9$}] {} --(1.5,0.5) node [midway,label={[label distance = 1pt]left: $\scriptstyle10$}] {} --(2.5,0.5) node [midway,label={[label distance = 1pt]below: $\scriptstyle11$}] {} --(2.5,-0.5) node [midway,label={[label distance = 1pt]left: $\scriptstyle12$}] {} --(3.5,-0.5) node [midway,label={[label distance = 1pt]below: $\scriptstyle13$}] {} --(3.5,0.5) node [midway,label={[label distance = 1pt]left: $\scriptstyle14$}] {} --(4.5,0.5) node [midway,label={[label distance = 1pt]below: $\scriptstyle15$}] {};
	\draw (1.5,3.5) node[face] {};
	\draw (3.5,0.5) node[face] {};
	\draw (3.5,-0.5) node[face] {};
	\draw (0.5,2.5) node[face] {};
	\draw (2.5,1.5) node[face] {};
	\draw (0.5,3.5) node[face] {};
	\draw (-0.5,3.5) node[face] {};
	\draw (1.5,2.5) node[face] {};
	\draw (2.5,2.5) node[face] {};
	\draw (1.5,1.5) node[face] {};
	\draw (3.5,2.5) node[face] {};
	\draw (-0.5,-0.5) node[face] {};
	\draw (-0.5,0.5) node[face] {};
	\draw (0.5,0.5) node[face] {};
	\draw (0.5,-0.5) node[face] {};
	\draw (-1.5,2.5) node[face] {};
	\draw (1.5,0.5) node[face] {};
	\draw (1.5,-0.5) node[face] {};
	\draw (0.5,1.5) node[face] {};
	\draw (-0.5,1.5) node[face] {};
	\draw (3.5,1.5) node[face] {};
	\draw (2.5,0.5) node[face] {};
	\draw (2.5,-0.5) node[face] {};
	\draw (-0.5,2.5) node[face] {};
	\draw (-1.5,1.5) node[face] {};
	\draw (-1.5,1.5) node[dobFace] {};
	\draw (3.5,0.5) node[dobFace] {};
\end{tikzpicture}
\def\exampleEtime{15}
\caption{The edges forming the Peierls interface $\gamma$ associated with the contour configuration 
in \cref{fig:P_example} and a compatible choice of dual edges $\bedge,\eedge$, 
	shown in blue and numbered according to the construction in \cref{def:interface}, with ambiguities resolved according to the  North-East/South-West rule. 
Observe that the numbering is reversed if the roles of $\bface$ and $\eface$ are reversed.  
In this example, we have $\etime=\exampleEtime$.
}
\label{fig:contour}
\end{figure}

Letting $\etime$ be the highest value of $n$, that is, the time when the exploration path terminates with $\gamma_{\etime} = \eedge$, 
the sequence $\gamma_{[0,n]} = (\gamma_0,\gamma_1,\ldots,\gamma_{n\wedge\etime})$ is a Markov chain and $\etime$ is an associated stopping time (for the natural filtration generated by $\gamma$).  
Note that the interface is manifestly reversible: 
exchanging $\bhalfedge$ and $\ehalfedge$ gives two different Markov chains, which however are related by time-reversal with respect to $\etime$ as in~\cite{Sheffield-Sun:Strong_path_convergence_from_Loewner_driving_function_convergence}. In particular, the terminal time $\etime$ and the path 
$\{ \gamma_1,\ldots,\gamma_{\etime} \}$ 
are unchanged as random variables.  
The state space of the Markov chain is the collection of all sequences $(\gamma_0,\gamma_1,\ldots,\gamma_n)$ of dual edges of various lengths $n \in \bZnn$.
As the graph is assumed to be bounded, the 
length is bounded from above (so the state space is finite but possibly large).

Denote $\PowerSet{\cEdual\Omega} = \{ X \; | \; X \subset \cEdual\Omega \}$.
Given a family of complex-valued functions $W^\xi \colon \smash{\cP^\xi_\Omega} \to \bC$ parameterized by $\xi \in \cVcluster\Omega \times \cVcluster\Omega$, for each $n \in \bZnn$,  
define a family of functions 
\begin{align}
\nonumber
F_W(\gamma_{[0,n]};\cdot) \colon \; & \cVcluster\Omega \to \bC, \\ 
F_W(\gamma_{[0,n]};\halfedge)
:= \; & 
\sum_{P \in \cP_\Omega^{\bhalfedge,\halfedge}} 
\one_{C_{\gamma,n}}(P) 
\, W^{\bhalfedge,\halfedge}(P) ,
\qquad \halfedge = (f,\edual) \in \cVcluster\Omega , 
\label{eq:F_def}
\end{align}
with $\bhalfedge = \vec \gamma_0$,
where $C_{\gamma,n}$ is the event that $\gamma_{[0,n]}$ is part of the Peierls interface, i.e., 
\begin{align*}
C_{\gamma,n} := \{ P \subset \cEdual\Omega \;|\; \gamma_{[0,n]} = \gamma_{[0,n]}(P) \} .
\end{align*}
A contour configuration $P \in C_{\gamma,n}$ necessarily contains $\{ \gamma_1,\ldots,\gamma_n \}$, 
but also excludes some adjacent dual edges which would have formed part of the interface if they were present. 
More precisely, for $n \le \etime$, from the construction of $\gamma_{[0,n]}(P)$ and the parity constraint in the definition of $\smash{\cP^\xi_\Omega}$, 
it follows that $C_{\gamma,n}\cup \smash{\cP^\xi_\Omega} $ consists exactly of those configurations satisfying both of the following properties:
\begin{enumerate}
\item $\gamma_t \in P$ for all $1 \le t \le n \vee (\etime-1)$; 

\item $e \notin P$ for all other edges $e$ which share a vertex with two edges in the path $\gamma_{[0,n]}$ 
that do not form a NE/SW pair 
(such as the edges drawn in red in \cref{fig:ContourConstraint}) 
--- indeed, such edges would have taken priority according to the NE/SW rule.
\end{enumerate}
\noindent 
In particular, if these two requirements are contradictory or the second is impossible, $\gamma_{[0,n]}$ does not respect the NE/SW rule and $C_{\gamma,n}$ is empty.  
Aside from this degenerate case, the set of edges which are not constrained coincides with $\cEdual{\Omega \setminus \tilde \gamma_{[0,n]}}$, where $\tilde \gamma_{[0,n]}$ is the set of vertices $v \in \cVprimal{}$ adjacent to at least one of the fixed edges --- see \cref{fig:ContourConstraint}.

\begin{figure}[h!]
\centering
\tikzpicturedependsonfile{SpinConstraint.tikz}
\begin{tikzpicture}
	\draw[dualEdge] (3.5,0.5) -- (3.5,1.5);
	\draw[dualEdge] (3.5,-0.5) -- (3.5,0.5);
	\draw[dualEdge] (0.5,2.5) -- (0.5,3.5);
	\draw[dualEdge] (0.5,2.5) -- (1.5,2.5);
	\draw[dualEdge] (2.5,1.5) -- (2.5,2.5);
	\draw[dualEdge] (2.5,1.5) -- (3.5,1.5);
	\draw[dualEdge] (0.5,3.5) -- (1.5,3.5);
	\draw[dualEdge] (-0.5,3.5) -- (0.5,3.5);
	\draw[dualEdge] (1.5,2.5) -- (1.5,3.5);
	\draw[dualEdge] (2.5,2.5) -- (3.5,2.5);
	\draw[dualEdge] (1.5,1.5) -- (2.5,1.5);
	\draw[dualEdge] (1.5,1.5) -- (1.5,2.5);
	\draw[dualEdge] (-0.5,-0.5) -- (-0.5,0.5);
	\draw[dualEdge] (-0.5,-0.5) -- (0.5,-0.5);
	\draw[dualEdge] (-0.5,0.5) -- (0.5,0.5);
	\draw[dualEdge] (-0.5,0.5) -- (-0.5,1.5);
	\draw[dualEdge] (0.5,0.5) -- (1.5,0.5);
	\draw[dualEdge] (0.5,0.5) -- (0.5,1.5);
	\draw[dualEdge] (0.5,-0.5) -- (0.5,0.5);
	\draw[dualEdge] (0.5,-0.5) -- (1.5,-0.5);
	\draw[dualEdge] (-1.5,2.5) -- (-0.5,2.5);
	\draw[dualEdge] (1.5,0.5) -- (1.5,1.5);
	\draw[dualEdge] (1.5,0.5) -- (2.5,0.5);
	\draw[dualEdge] (1.5,-0.5) -- (1.5,0.5);
	\draw[dualEdge] (0.5,1.5) -- (0.5,2.5);
	\draw[dualEdge] (0.5,1.5) -- (1.5,1.5);
	\draw[dualEdge] (-0.5,1.5) -- (0.5,1.5);
	\draw[dualEdge] (-0.5,1.5) -- (-0.5,2.5);
	\draw[dualEdge] (3.5,1.5) -- (3.5,2.5);
	\draw[dualEdge] (2.5,0.5) -- (3.5,0.5);
	\draw[dualEdge] (2.5,0.5) -- (2.5,1.5);
	\draw[dualEdge] (2.5,-0.5) -- (3.5,-0.5);
	\draw[dualEdge] (2.5,-0.5) -- (2.5,0.5);
	\draw[dualEdge] (-0.5,2.5) -- (0.5,2.5);
	\draw[dualEdge] (-0.5,2.5) -- (-0.5,3.5);
	\draw[dualEdge] (-1.5,1.5) -- (-1.5,2.5);
	\draw[dualEdge] (-1.5,1.5) -- (-0.5,1.5);
	\draw[peierls] (-1.5,1.5) --(-0.5,1.5) --(-0.5,0.5) --(0.5,0.5) --(0.5,1.5) --(-0.5,1.5) --(-0.5,2.5) --(0.5,2.5) --(0.5,1.5) --(1.5,1.5);
	\draw[excludedEdge] (-0.5,2.5) -- (-1.5,2.5);
	\draw[excludedEdge] (0.5,0.5) -- (1.5,0.5);
	\draw[excludedEdge] (-0.5,2.5) -- (-0.5,3.5);
	\draw[excludedEdge] (0.5,0.5) -- (0.5,-0.5);
	\draw (1.5,3.5) node[face] {};
	\draw (3.5,0.5) node[face] {};
	\draw (3.5,-0.5) node[face] {};
	\draw (0.5,2.5) node[face] {};
	\draw (2.5,1.5) node[face] {};
	\draw (0.5,3.5) node[face] {};
	\draw (-0.5,3.5) node[face] {};
	\draw (1.5,2.5) node[face] {};
	\draw (2.5,2.5) node[face] {};
	\draw (1.5,1.5) node[face] {};
	\draw (3.5,2.5) node[face] {};
	\draw (-0.5,-0.5) node[face] {};
	\draw (-0.5,0.5) node[face] {};
	\draw (0.5,0.5) node[face] {};
	\draw (0.5,-0.5) node[face] {};
	\draw (-1.5,2.5) node[face] {};
	\draw (1.5,0.5) node[face] {};
	\draw (1.5,-0.5) node[face] {};
	\draw (0.5,1.5) node[face] {};
	\draw (-0.5,1.5) node[face] {};
	\draw (3.5,1.5) node[face] {};
	\draw (2.5,0.5) node[face] {};
	\draw (2.5,-0.5) node[face] {};
	\draw (-0.5,2.5) node[face] {};
	\draw (-1.5,1.5) node[face] {};
	\draw (-1.5,1.5) node[dobFace] {};
	\draw (3.5,0.5) node[dobFace] {};
	\draw (-1,1) node[edgeVertexSpinMinus] {};
	\draw (0,0) node[edgeVertexSpinMinus] {};
	\draw (0,2) node[edgeVertexSpinMinus] {};
	\draw (1,1) node[edgeVertexSpinMinus] {};
	\draw (1,0) node[edgeVertexSpinMinus] {};
	\draw (-1,2) node[edgeVertexSpinPlus] {};
	\draw (-1,3) node[edgeVertexSpinPlus] {};
	\draw (0,3) node[edgeVertexSpinPlus] {};
	\draw (0,1) node[edgeVertexSpinPlus] {};
	\draw (1,2) node[edgeVertexSpinPlus] {};
\end{tikzpicture}
\caption{An example of $\tilde \gamma_{[0,n]}$ with the same $\gamma$ as in the previous figures and $n = \examplen$.  
The edges which are always included (resp.~excluded) for all $P \in C_{\gamma,n}$ are drawn in blue (resp.~red).
The set $\tilde \gamma_{[0,n]}$ can also be understood as the set of vertices whose relative spin is fixed by the presence of the interface up to time $n$, 
and the vertices in $\tilde \gamma_{[0,n]}$ (not necessarily in $\Omega$) are shown as 
``\tikzvertex{edgeVertexSpinPlus}'' or ``\tikzvertex{edgeVertexSpinMinus}'' 
according to how their relative spin is fixed.
}
\label{fig:ContourConstraint}
\end{figure}

As a result, we make the following observation, which will be useful in \cref{sec:Grassmann_martingale}. 

\begin{lem}
\label{lem:Cgamma}
For each Peierls interface $\gamma$ \textnormal{(}Def.~\ref{def:interface}\textnormal{)}, integer $n$, 
and boundary condition 
$\xi=(\bhalfedge,\halfedge)$, 
there exists a set $\tilde \gamma_{[0,n]} \subset \cVprimal{}$ depending only on $\gamma_{[0,n]}$ 
such that the assignment
\begin{align}
P = \hat P \cup \left\{ \gamma_1,\ldots,\gamma_n \right\}
\label{eq:Cgamma_P_reduce}
\end{align}
is a one-to-one correspondence between $P \in \smash{\cP^\xi_\Omega}$ 
and $\hat P \in \cP^{\vec\gamma_n,\halfedge}_{\Omega \setminus \tilde \gamma_{[0,n]}}$.
\end{lem}

Now, assume that for each $\ehalfedge$ such that $\xi=(\bhalfedge,\ehalfedge)$ is admissible, 
the argument of $W$ is constant on $\cP_\Omega^{\xi}$. 
Then, we may define a family of probability measures 
\begin{align}
\bP_W^{\vec \gamma_0,\ehalfedge}[\gamma_{[0,n]}]
:=
\frac{F_W(\gamma_{[0,n]};\ehalfedge)}{F_W(\gamma_{[0,0]};\ehalfedge)} , \qquad n \in \bZnn .
\label{eq:generic_prob}
\end{align}
For example, when $|W(P)| = \prod_{\edual \in P} x_{\edual}$, 
this is the distribution of the interface associated with the Ising model with Dobrushin boundary conditions,
as defined in \cref{sec:setup}.

Standard first-step analysis of the interface shows that any function of the form~\eqref{eq:F_def} gives rise to a martingale\footnote{Here, we do not require $\xi$ to be admissible --- while for relating this to the spin-Ising model we should, for models on contour configurations this makes sense for any boundary condition.}  (up to a suitable stopping time). 
This observation (due to Doob) lies at the heart of all SLE convergence proofs: conditional expectations of a random variable given an increasing sequence of sigma-algebras form a tautological martingale.

\begin{prop}
\label{prop:M_general}
Fix $\halfedge = (f,\edual) \in \cVcluster\Omega$ 
and boundary condition $\xi = (\bhalfedge,\ehalfedge)$. 
The process
\begin{align}
M_W( \halfedge ) [\gamma_{[0,n]}]
:=
\frac{F_W(\gamma_{[0,n]};\halfedge)}{F_W(\gamma_{[0,n]};\ehalfedge)} , \qquad n \leq \etime ,
\label{eq:M_def}
\end{align}
is a local martingale with respect to the natural filtration generated by $\gamma$.  
\end{prop}

\begin{proof}
Let $n < \etime$.
Note that the $(n+1)$st step of $\gamma$ satisfies
\begin{align}
\bP_W^{\xi} \big[ \gamma_{n+1} \big| \gamma_{[0,n]} \big]
=
\frac{F_W(\gamma_{[0,n+1]};\ehalfedge)}{F_W(\gamma_{[0,n]};\ehalfedge)} .
\label{eq:generic_conditional}
\end{align}
Given $\gamma_{[0,n]}$, consider the possible dual edges $\edual_+ \in \cEdual{}$ for which $\gamma+ := \gamma_{[0,n]} \cup \{ \edual_+ \}$ has positive probability (i.e., the possible continuations of the Peierls interface).  
As the different events $C_{\gamma+,n+1}$ 
are disjoint and their union is $C_{\gamma,n}$, 
the definition~\eqref{eq:F_def} of $F_W$ gives
\begin{align}
\begin{split}
\sum_{\edual_+} 
F_W(\gamma+;\halfedge) \;
= \; &
\sum_{P \in \cP_\Omega^{\varpi}}
\sum_{\edual_+} 
\one_{C_{\gamma+,n+1}}(P) 
\, W^\varpi(P)
\\ 
= \; &
\sum_{P \in \cP_\Omega^{\varpi}}
\one_{C_{\gamma,n}}(P) 
\, W^\varpi(P)
=
F_W(\gamma_{[0,n]};\halfedge) ,
\end{split}
\label{eq:F_sumrule}
\end{align}
where $\varpi = (\bhalfedge,\halfedge)$. 
Combining \cref{eq:F_sumrule}  with \cref{eq:generic_conditional} yields
\begin{align}
\label{eq:martingale_proof}
\bE_W^{\xi} \Big[ M_W(\halfedge)[\gamma_{[0,n+1]}] \; \Big| \; \gamma_{[0,n]} \Big]
\underset{\hphantom{\eqref{eq:generic_conditional}}}{\overset{\hphantom{\eqref{eq:F_sumrule}}}{=}} \; &
\sum_{\edual_+}
\bP_W^{\xi} \big[\gamma_{n+1}=\edual_+ \; \big| \;  \gamma_{[0,n]} \big]
\; \frac{F_W(\gamma+ ; \halfedge)}{F_W(\gamma+;\ehalfedge)}
\\ 
\underset{\hphantom{\eqref{eq:F_sumrule}}}{\overset{\eqref{eq:generic_conditional}}{=}} \; &
\frac{
1
}{F_W(\gamma_{[0,n]};\ehalfedge)}
\; \sum_{\edual_+}
F_W(\gamma+ ; \halfedge)
\; 
\underset{\hphantom{\eqref{eq:generic_conditional}}}{\overset{\eqref{eq:F_sumrule}}{=}} \; 
M_W(\halfedge)[\gamma_{[0,n]}] .
\nonumber
\end{align}
This shows the martingale property for one step, which is sufficient to conclude.
\end{proof}

Note that even for a fixed distribution of the interface process there are many possible functions $W$ with the needed properties, defining different martingale observables.
For example, the original one used by Smirnov~\cite{Smirnov:Towards_conformal_invariance_of_2D_lattice_models}
and 
Chelkak~\&~Smirnov~\cite{Chelkak-Smirnov:Universality_in_2D_Ising_and_conformal_invariance_of_fermionic_observables},
further elaborated especially by Izyurov~\cite{Izyurov:Critical_Ising_interfaces_in_multiply_connected_domains}, 
is a linear combination with two different values of $\halfedge$ sharing the same edge, 
as well as involving a specific choice of the functions $W$ with a sign depending on the winding of the contours.
Appropriate choices are important to actually get a meromorphic scaling limit via discrete complex analysis techniques.

\bigskip{}
\section{Grassmann integral representation of the planar Ising model}
\label{sec:Grassman}

We now review the representation of the 
(integrable) Ising model in terms of Grassmann variables, in particular introducing 
a specific choice of the functions $W$ for~\eqref{eq:F_def} which produces a nice version of the martingale observable used to study the interface process.  

First, in \cref{sec:Grassmann_preli} we summarize a representation of the planar Ising model in terms of Grassmann calculus.
There are a number of slightly different versions of this representation and generalizations to other planar graphs, for which we refer the reader 
to~\cite{CCK:Revisiting_the_combinatorics_of_the_2D_Ising_model} and references therein.  
Note, however, that the representation that we use 
(essentially the one which is most common in the renormalization group literature, 
cf.~\cite{GGM:The_scaling_limit_of_the_energy_correlations_in_non_integrable_Ising_models} and references therein) 
is not quite the same as that discussed in~\cite{CCK:Revisiting_the_combinatorics_of_the_2D_Ising_model}
--- most importantly, regarding positions of the factors of $x_e$ (compare \cref{eq:cS_def} to~\cite[Equation~(1.7)]{CCK:Revisiting_the_combinatorics_of_the_2D_Ising_model}).
Consequently, the representation of~\cite{CCK:Revisiting_the_combinatorics_of_the_2D_Ising_model} 
has a somewhat less direct relationship to the contour representation, making it less convenient for our purposes.

Then, in \cref{sec:Grassman_Dobrushin} we analyze the effect of boundary conditions relevant to the non-integrable generalization of the planar Ising model. 
In \cref{sec:Grassmann_martingale}, we will end up with an expression which (in the integrable case) is already known to have a suitable scaling limit, 
in terms of correlation functions of a lattice fermionic quantum field theory (\cref{prop:M_free_expectation}).
It is important to observe that this Grassmann observable is local.

\subsection{Grassmann field, aka lattice fermion}
\label{sec:Grassmann_preli}

Fix finite $\Omega \subset \cVprimal{}$ and consider the cluster graph $\cGcluster\Omega = ( \cVcluster\Omega, \cEcluster\Omega )$ associated to $\Omega$. 
Let $\Phi_\Omega$ be the complex \emph{Grassmann algebra}\footnote{That is, $\Phi_\Omega$ is isomorphic to the exterior algebra of the vector space $V_\Omega := \Span_\bC\{\varphi_{\halfedge} \;|\; \halfedge \in \cVcluster\Omega\}$ 
obtained as the quotient of the tensor algebra of $V_\Omega$ by the relation $\varphi_{\halfedge} \varphi_{\halfedge'} = - \varphi_{\halfedge'} \varphi_{\halfedge}$ for all $\halfedge, \halfedge' \in \cVcluster\Omega$.} 
with a basis $\varphi_{\halfedge}$ indexed by $\halfedge \in \cVcluster\Omega$. 
For each face $f \in \cFprimal\Omega$ we label the vertices of the associated complete four-graph $\diamondplus_f$ as  
$\vcluster{f}{N}, \vcluster{f}{E}, \vcluster{f}{S}, \vcluster{f}{W} \in \cVcluster\Omega$, according to the direction of the associated edge (North, East, South, West) as seen from $f$.
For brevity, we also denote the four associated Grassmann basis variables by 
\begin{align*}
\GrN_f := \varphi_{\vcluster{f}{N}}, \qquad 
\GrE_f := \varphi_{\vcluster{f}{E}}, \qquad 
\GrS_f := \varphi_{\vcluster{f}{S}}, \qquad 
\GrW_f := \varphi_{\vcluster{f}{W}}  ,
\end{align*}
the subalgebra they generate by $\Phi_f$, and the associated Grassmann (Berezin) integral (see, e.g.,~\cite[Appendix~2.B]{DMS:CFT}) by
\begin{align}
\int \cD [\Phi_f] \ (\cdot) := \int \ud \GrW_f \ud \GrS_f \ud \GrE_f \ud \GrN_f  \ (\cdot) 
\label{eq:one_face_integration}
\end{align}
(note the possibly counterintuitive ordering which is chosen to simplify a sign later).

The Grassmann action functional is given by the interaction of Grassmann basis variables across all short edges. The action functional around one face $f \in \cFprimal\Omega$ is 
\begin{align}\label{eq:cS0_def}
\begin{split}
\cS_0 (\Phi_f)
:= 
\GrE_f \GrW_f
+
\GrN_f \GrS_f
+
\GrN_f \GrE_f 
+
\GrS_f \GrW_f
+
\GrS_f \GrE_f
+
\GrW_f \GrN_f
.
\end{split}
\end{align}
Recalling that anticommutation implies that squares of Grassmann variables equal zero, and so
$\exp ( \cS_0 (\Phi_f) ) = 1 + \cS_0 (\Phi_f)$,
and noting that any Grassmann integral not involving all of the Grassmann variables exactly once vanishes, we obtain
\begin{align}
\int \cD [\Phi_f] \ 
e^{\cS_0 (\Phi_f)}
&=1
=
\int \cD [\Phi_f] \
\GrN_f \, \GrE_f 
\, \GrS_f \, \GrW_f ,
\label{eq:exp_cS0}
\\
\int \cD[\Phi_f] \ 
\varphi^N_f \varphi^E_f \varphi^S_f \varphi^W_f 
e^{\cS_0(\Phi_f)}
&=
1
=
\int \cD[\Phi_f] \ 
\varphi^N_f \varphi^E_f  
e^{\cS_0(\Phi_f)}
\int \cD[\Phi_f] \ 
\varphi^S_f \varphi^W_f 
e^{\cS_0(\Phi_f)} 
\label{eq:face_products_cS0}
\end{align}
(cf.~\cite[Equations~(2.219)~\&~(2.225)]{DMS:CFT}). 
The interaction across long edges is given~by
\begin{align}
E_{\edual} :=
\begin{cases}
\GrE_f \, \GrW_{f+1}
, &
\edual = \{ f,f+1 \} \subset \cEdual\Omega ,
\\
\GrN_f \, \GrS_{f+\ii}
, &
\edual = \{ f,f+\ii \} \subset \cEdual\Omega ,
\end{cases}
\label{eq:Ee_def}
\end{align}
where $f$ denotes the midpoint of the face $f$, i.e., a point in dual lattice $(\bZ + \tfrac12) + \ii ( \bZ + \tfrac12 )$.
We use this to define a Grassmann polynomial on $\Omega$, sometimes called an action functional, 
\begin{align}
\cS_{\ul x} (\Phi_\Omega)
:=
\sum_{\edual \in \cEdual\Omega} E_{\edual} \, x_{\edual} 
+
\sum_{f \in \cFprimal\Omega} \cS_0 (\Phi_f)
,
\label{eq:cS_def}
\end{align}
where $\ul x = (x_{\edual})_{\edual \in \cEdual\Omega} \in [0,\infty)^{\cEdual\Omega}$ are arbitrary edge weights. 
We define a Grassmann integral over the entire algebra  via~\eqref{eq:one_face_integration} by
\begin{align*}
\int \cD [\Phi_\Omega] \ (\cdot) 
:= \int \prod_{f \in \cFprimal\Omega} \cD [\Phi_f] \ (\cdot) 
= \int \prod_{f \in \cFprimal\Omega} \ud \GrW_f \ud \GrS_f \ud \GrE_f \ud \GrN_f \ (\cdot)
\end{align*}
(note that the order of the faces is not important). 
The key feature of this object is that it reproduces the partition function with fixed (plus) boundary conditions as
\begin{align}
Z_{\Omega,\ul x}^{+}
= \int \cD [\Phi_\Omega] \ 
e^{\cS_{\ul x}(\Phi_\Omega)} ,
\label{eq:Z_plus_Grassman}
\end{align}
since the integral is given by the same Pfaffian as in Kasteleyn's solution of the Ising model~\cite{Kasteleyn:Dimer_statistics_and_phase_transitions} 
(see also~\cite{Samuel:The_use_of_anticommuting_variable_integrals_in_statistical_mechanics}, and~\cite{CCK:Revisiting_the_combinatorics_of_the_2D_Ising_model} for the case with generic edge weights)\footnote{\Cref{eq:Z_plus_Grassman} and similar formulae often appear with an ambiguous sign, but the fact that the signs of the two expressions match can be checked by setting $\ul x = 0$ and using \cref{eq:exp_cS0}.}. 
This is closely related to the following identity, which will be the key to relating these integrals to other quantities of interest:

\begin{lem}
\label{lem:Grass_int_inndicator}
For any $X \subset \cEdual\Omega$, we have 
\begin{align}
\label{eq:Grass_int_inndicator}
\int \cD [ \Phi_\Omega ] 
\Big(
\prod_{\edual \in X} E_{\edual} 
\Big) \,
e^{\cS_{0} (\Phi_\Omega)} 
= 
\one_{\cP_\Omega}(X) 
.
\end{align}
\end{lem}
Note that this is in fact equivalent to one of the intermediate steps in Kasteleyn's solution, but rather than entering into the technical details needed to verify this, 
we instead derive it from the fact that \cref{eq:Z_plus_Grassman} holds for all values of the edge weights.
\begin{proof}
From \cref{eq:Z_plus_lt}, we obtain
\begin{align*}
Z_{\Omega,\ul x}^{+}
=
\sum_{P \in \cP_\Omega}
\prod_{e \in P} x_e
=
\sum_{X \subset \cEdual\Omega}
\one_{\cP_\Omega}(X) 
\prod_{e \in X} x_e .
\end{align*}
On the other hand, using \cref{eq:Z_plus_Grassman} and expanding the first term in the exponential in \cref{eq:cS_def} 
as a product of terms of the form 
$\exp ( x_{\edual} E_{\edual}) = 1 + x_{\edual} E_{\edual}$,  we obtain
\begin{align*}
Z_{\Omega,\ul x}^{+}
=
\int \cD [\Phi_\Omega] \ 
e^{\cS_{\ul x}(\Phi_\Omega)}
= \; & \sum_{X \subset \cEdual\Omega} \Big( \prod_{\edual \in X} x_{\edual} \Big) \, 
\int \cD [ \Phi_\Omega ] 
\Big(
\prod_{\edual \in X} E_{\edual} 
\Big) \,
e^{\cS_{0} (\Phi_\Omega)} .
\end{align*}
Since the edge weights may be arbitrary, these two expressions are equal as polynomials, and equating the coefficients gives \cref{eq:Grass_int_inndicator}.
\end{proof}

We will find it convenient to define a Berezin measure related to the expression for the partition function in \cref{eq:Z_plus_Grassman} by 
\begin{align}
\label{eq:correlation_function}
\big\langle \cO \big\rangle_{\Omega,\ul x}^{+} 
:= 
\frac{\int \cD [\Phi_\Omega] \ \cO(\Phi_\Omega) \ e^{\cS_{\ul x} (\Phi_\Omega)}}{\int \cD [\Phi_\Omega] \ e^{\cS_{\ul x} (\Phi_\Omega)}}
= 
\frac{1}{Z_{\Omega,\ul x}^{+}} \, \int \cD [\Phi_\Omega]\ \cO(\Phi_\Omega) \ e^{\cS_{\ul x} (\Phi_\Omega)}  ,
\end{align}
where $\cO = \cO(\Phi_\Omega)$ is a function of the Grassmann variables 
(i.e., a generic element of the Grassmann algebra).

\subsection{Grassmann integrals and boundary conditions}
\label{sec:Grassman_Dobrushin}

As we will now show, similar expressions are available\footnote{This could also be done in terms of the relationship to the dimer model on the cluster graph, but we will instead give a derivation in terms of manipulations of Grassmann integrals, which also produces some formulas which will be helpful later.} 
for the sets of non-even contours involved in Dobrushin and alternating boundary conditions.  
Given a finite set $\Omega \subset \cVprimal{}$, an integer $N \in \bZnn$, and distinct vertices (half-edges) $\xi = (\halfedge_1, \halfedge_2, \ldots, \halfedge_{2N})$ in $\cVcluster\Omega$, let
\begin{align}
W_{\Omega,\ul x}^{\xi}(X) :=
\int \cD [ \Phi_\Omega ] \, \varphi_{\halfedge_1} \, \varphi_{\halfedge_2} \, \cdots \, \varphi_{\halfedge_{2N}}
\Big(
\prod_{\edual \in X} x_{\edual} E_{\edual} 
\Big) \,
e^{\cS_{0} (\Phi_\Omega)} , \qquad X \subset \cEdual\Omega .
\label{eq:W_nonint}
\end{align}
In \cref{prop:part_f_equals_Ber_int}, we relate similar Grassmann integrals to the partition function $Z_{\Omega,\ul x}^{\xi}$.

\begin{lem}
\label{lem:cP_sum_Grassman_general}
For any $X \subset \cEdual\Omega$, we have
\begin{align}
\begin{split}
\sum_{\substack{P \in \cP_\Omega^{\xi} \\ P \supset X}} 
W_{\Omega,\ul x}^{\xi}(P)
= \; & 
\bigg( \prod_{\edual \in X} x_{\edual} \frac{\partial}{\partial x_{\edual}} \bigg)
\sum_{P \in \cP_\Omega^{\xi}} 
W_{\Omega,\ul x}^{\xi}(P) 
\\
= \; & 
\int \cD [ \Phi_\Omega ] \, \varphi_{\halfedge_1} \, \varphi_{\halfedge_2} \, \cdots \, \varphi_{\halfedge_{2N}}
\Big(
\prod_{\edual \in X} x_{\edual} E_{\edual} 
\Big) \,
e^{\cS_{\ul x} (\Phi_\Omega)} .
\end{split}
\label{eq:cP_sum_derivs}
\end{align}
\end{lem}

\begin{proof}
Fix $X \subset \cEdual\Omega$. 
The first line of~\eqref{eq:cP_sum_derivs} is evident from~\eqref{eq:W_nonint}.
It remains to prove the second line of~\eqref{eq:cP_sum_derivs}.
By expanding the first term in the  exponential in \cref{eq:cS_def} 
as a product of terms of the form 
$\exp ( x_{\edual} E_{\edual}) = 1 + x_{\edual} E_{\edual}$, 
and noting that inside the Grassmann integral 
all terms which do not correspond to contour configurations $P \in \cP_\Omega^{\xi}$ vanish since they cannot be written without repetitions of Grassmann variables, we obtain
\begin{align*}
 \; & \int \cD [ \Phi_\Omega ] \, \varphi_{\halfedge_1} \, \varphi_{\halfedge_2} \, \cdots \, \varphi_{\halfedge_{2N}}
\Big(
\prod_{\edual \in X} x_{\edual} E_{\edual} 
\Big) \,
e^{\cS_{\ul x} (\Phi_\Omega)} 
\\
= \; & \sum_{\substack{P \in \cP_\Omega^{\xi} \\ P \supset X}} 
\int \cD [ \Phi_\Omega ] \, \varphi_{\halfedge_1} \, \varphi_{\halfedge_2} \, \cdots \, \varphi_{\halfedge_{2N}}
\Big(
\prod_{\edual \in P} x_{\edual} E_{\edual} 
\Big) \,
e^{\cS_{0} (\Phi_\Omega)} .
\end{align*}
The claim now follows by recognizing the definition~\eqref{eq:W_nonint} on the second line.
\end{proof}

Among other things, this gives an almost explicit formula for $W_{\Omega,\ul x}$ which connects it to the distribution of contour configurations in the Ising model and to \cref{sec:interface}.
\begin{prop} 
\label{prop:Ber_W_admissible}
For any admissible boundary condition $\xi = (\halfedge_1, \halfedge_2, \ldots, \halfedge_{2N})$, we have 
\begin{align}\label{eq:W_equals_prod_of_edge_weights}
W_{\Omega,\ul x}^{\xi}(P) 
= \; & \pm \prod_{\edual \in P} x_{\edual} ,
\qquad \textnormal{for all } P \in \cP_\Omega^{\xi} ,
\end{align}
with the same sign for all $P \in \cP_\Omega^{\xi}$. 
\end{prop}

\begin{proof}
Since $\xi$ is admissible, there exist $\Ombc \supset \Omega$ and $\Pbc \subset \cEdual{\Ombc\setminus \Omega}$ with $|\Ombc| < \infty$  such that $P \cup \Pbc \in \cP^\emptyset_{\Ombc}$ for all $P \in \cP_\Omega^\xi$.  
Then evidently,
\begin{align}
\varphi_{\halfedge_1} \, \varphi_{\halfedge_2} \, \cdots \, \varphi_{\halfedge_{2N}} 
\; & =  
\pm
\int \Big( \prod_{f \in \cFprimal{\Ombc} \setminus \cFprimal\Omega} \cD [\Phi_f] \ e^{\cS_{0} (\Phi_f)} \Big) \, 
\Big( \prod_{\edual \in \Pbc} E_{\edual}  \Big) \,  .
\label{eq:marked_phis_as_integral}
\end{align}
Set $x_{\edual} = 1$ for dual edges $\edual \in \cEdual{\Ombc} \setminus \cEdual\Omega$.
Then, we obtain 
\begin{align*}
W_{\Omega,\ul x}^{\xi}(P) 
= \; & \int \cD [\Phi_\Omega] \
\varphi_{\halfedge_1} \, \varphi_{\halfedge_2} \, \cdots \, \varphi_{\halfedge_{2N}} \
\Big( \prod_{\edual \in P} x_{\edual} E_{\edual}  \Big) \ e^{\cS_{0} (\Phi_\Omega)}
&&\textnormal{[by~\eqref{eq:W_nonint}]}
\\
= \; & \pm
\int \cD [\Phi_{\Ombc}] \
\Big( \prod_{\edual \in \Pbc \cup P} x_{\edual} E_{\edual}  \Big) \ e^{\cS_{0} (\Phi_\Omega)}
\\
= \; & \pm \prod_{\edual \in  P} x_{\edual} ,
\qquad P \in \cP_\Omega^{\xi} , 
&&\textnormal{[by~\cref{lem:Grass_int_inndicator}]}
\end{align*}
where the sign on both lines equals that in~\eqref{eq:marked_phis_as_integral}, 
so it is indeed independent of ${P \in \cP_\Omega^{\xi}}$.
\end{proof}

\begin{remark}
Without assuming that $\xi$ is admissible, \cref{eq:W_equals_prod_of_edge_weights} still holds, 
but we lose control over the dependence of the sign on $P$.
\label{rem:W_sign}
\end{remark}

Combining this with the $X=\emptyset$ case of \cref{lem:cP_sum_Grassman_general} gives an expression for the partition function, up to a sign:

\begin{cor} 
\label{prop:part_f_equals_Ber_int}
For any admissible boundary condition $\xi = (\halfedge_1, \halfedge_2, \ldots, \halfedge_{2N})$,  
\begin{align}
\label{eq:part_f_equals_Ber_int}
Z_{\Omega,\ul x}^{\xi} 
= 
\Bigg|\int \cD [ \Phi_\Omega ] \, \varphi_{\halfedge_1} \, \varphi_{\halfedge_2} \, \cdots \, \varphi_{\halfedge_{2N}} \ e^{\cS_{\ul x} (\Phi_\Omega)} \Bigg| .
\end{align}
\end{cor}

In this case, unlike in the case of uniform boundary conditions, 
the Grassmann integral does not have a fixed sign; 
in particular it is odd under permutations of the order of the Grassmann variables inserted in the integral, 
and thus also depends nontrivially on the order of the elements of $\xi$, which is why we have defined it as an ordered tuple.

\subsection{Grassmann integral representation of the martingale observable}
\label{sec:Grassmann_martingale}

From \cref{prop:Ber_W_admissible}, we see that $W$ defined in \cref{eq:W_nonint} in \cref{sec:Grassman_Dobrushin} 
has the properties used in \cref{sec:interface} to construct a martingale observable~\eqref{eq:F_def} with respect to the exploration process.  
To be precise, for admissible boundary condition $\xi = (\bhalfedge,\ehalfedge)$, 
\begin{align}
\label{eq:F_def_again}
F_W(\gamma_{[0,n]};\halfedge)
:=
\sum_{P \in \cP_\Omega^{\bhalfedge,\halfedge}} 
\one_{C_{\gamma,n}}(P) 
\, W_{\Omega,\ul x}^{\bhalfedge,\halfedge}(P) ,
\end{align}
where $C_{\gamma,n} := \{ P \subset \cEdual\Omega \;|\; \gamma_{[0,n]} = \gamma_{[0,n]}(P) \}$, for each $h\in \cVcluster\Omega \setminus \{\bhalfedge\}$ the ratio~\eqref{eq:M_def}, 
\begin{align*}
M_W( \halfedge ) [\gamma_{[0,n]}]
:=
\frac{F_W(\gamma_{[0,n]};\halfedge)}{F_W(\gamma_{[0,n]};\ehalfedge)} 
, \qquad n \leq \etime ,
\end{align*}
is a (local) martingale with respect to the exploration process associated with 
the Gibbs measure for the Ising model with Dobrushin boundary conditions specified by $\xi$.

Using \cref{lem:Cgamma}, we obtain the following identity, which makes it possible to write $F$ more compactly as a Grassmann (Berezin) integral (expressed in \cref{prop:M_free_expectation}).

\begin{lem}
\label{lem:What_splitting}
For each half-edge $\halfedge \in \cVcluster\Omega \setminus \{\bhalfedge\}$, each Peierls interface $\gamma$ \textnormal{(}Def.~\ref{def:interface}\textnormal{)} with $\vec \gamma_0 := \bhalfedge$, and each integer $n$, 
there is a one-to-one correspondence between contours $P \in \cP_\Omega^{\vec \gamma_0,\halfedge} \cap C_{\gamma,n}$ and $\hat P(n) \in \cP^{\vec \gamma_n,\halfedge}$ such that 
\begin{align}
\label{eq:What_splitting}
W_{\Omega,\ul x}^{\vec \gamma_0,\halfedge}(P) 
=
\pm
\Big(
\prod_{k=1}^{n} x_{\gamma_k} 
\Big) 
\; W_{\Omega\setminus \tilde \gamma_{[0,n]},\ul x}^{\vec \gamma_n,\halfedge}(\hat P(n)) ,
\end{align}
where $\tilde \gamma_{[0,n]}$ is the set in \cref{lem:Cgamma}, 
and the sign is determined by the interface up to time $n$\textnormal{;} 
in particular, it is independent of $\halfedge$ and $P$. 
\end{lem}

\begin{proof}
\cref{lem:Cgamma} gives a unique contour configuration $\hat P(n)$ such that $P = \hat P(n) \cup \{ \gamma_1,\ldots,\gamma_n \}$.  
Splitting $E_{\gamma_n}$ from \eqref{eq:Ee_def} 
as a product of $\varphi_{(f(\vec \gamma_{n-1}),\gamma_n)}$ and $\varphi_{\vec \gamma_n}$, 
we obtain
\begin{align}
\label{eq:What}
\; & 
W_{\Omega,\ul x}^{\vec \gamma_0,\halfedge}(P) 
\\
\nonumber 
= \; & 
\int \cD [ \Phi_\Omega ] \, 
\varphi_{\vec \gamma_0} \,
\varphi_{\halfedge}
\Big(
\prod_{\edual \in \hat P(n) \cup \{ \gamma_1,\ldots,\gamma_n \}} x_{\edual} E_{\edual} 
\Big) \,
e^{\cS_{0} (\Phi_\Omega)} 
\\
\nonumber
= \; & 
\pm
\int \cD [ \Phi_\Omega ] \, 
\underbrace{\bigg(
\varphi_{\vec \gamma_0} \,
\Big(
\prod_{k=1}^{n-1} x_{\gamma_k}  E_{\gamma_k} 
\Big) \,
x_{\gamma_n}
\varphi_{(f(\vec \gamma_{n-1}),\gamma_n)}
\bigg)}_{(A)}
\underbrace{\bigg(
\varphi_{\vec \gamma_n}
\varphi_{\halfedge}
\Big(
\prod_{\edual \in \hat P(n)} x_{\edual} E_{\edual} 
\Big)
\bigg)}_{(B)} \; 
e^{\cS_{0} (\Phi_{\Omega })}  
,
\end{align}
with the same sign $\pm$ as in 
$E_{\gamma_n } = \pm \varphi_{(f(\vec \gamma_{n-1}),\gamma_n)} \varphi_{\vec \gamma_n}$, which depends only on $\gamma_{[0,n]}$.

We will separate the two parts (A)~\&~(B) in the integral into two disjoint Grassmann integrals. Such a separation is not completely trivial, since 
the sets $\cFprimal{\Omega \setminus \tilde \gamma_{[0,n]}}$ and $\cFprimal{\tilde \gamma_{[0,n]}}$ of faces respectively incident to the dual edges in 
$\cEdual{\Omega \setminus \tilde \gamma_{[0,n]}}$ and $\cEdual{\tilde \gamma_{[0,n]}}$
have a non-trivial intersection, consisting of faces $f$ which we divide into three types (see \cref{fig:ContourConstraint}): 
\begin{enumerate}
\item 
\label{item:face_case_1}
$f$ that are not incident to dual edges belonging to the Peierls interface $\gamma_{[0,n]}$;
\item 
\label{item:face_case_2}
$f$ that are incident to $\gamma_{[0,n]}$ but not incident to $\hat P(n)$; 
\item 
\label{item:face_case_3}
$f$ incident both to dual edges belonging to the Peierls interface $\gamma_{[0,n]}$ and to $\hat P(n)$.
\end{enumerate}
We can duplicate the integrals over these faces as follows.  

In Cases~\ref{item:face_case_1} and~\ref{item:face_case_2}, we may insert a factor~\eqref{eq:cS0_def} equaling one to the Grassmann integral associated to the corresponding face:
\begin{align*}
1= \int \cD [\Phi_f] \ e^{\cS_0 (\Phi_f)}
\end{align*}
which is grouped with (A) in Case~\ref{item:face_case_1} and with (B) in Case~\ref{item:face_case_2}.

In Case~\ref{item:face_case_3}, the parity constraints and the disambiguation rule for the Peierls interface (resolving the self-touchings of the interface in the North-East and South-West directions)
guarantee that the relevant part of the integral contains all four associated Grassmann variables, and so has the form
(using \cref{eq:face_products_cS0}, up to a reordering which does not change the sign of the final result)
\begin{align*}
\int \cD[\Phi_f] \ 
\varphi^N_f \varphi^E_f \varphi^S_f \varphi^W_f 
e^{\cS_0(\Phi_f)}
= 1
= 
\int \cD[\Phi_f] \ 
\varphi^N_f \varphi^E_f  
e^{\cS_0(\Phi_f)}
\int \cD[\Phi_f] \ 
 \varphi^S_f \varphi^W_f 
e^{\cS_0(\Phi_f)}  .
\end{align*}
It now follows that 
\begin{align}
\begin{split}
\textnormal{\eqref{eq:What}}
= \; & \pm 
\Big(
\prod_{k=1}^{n} x_{\gamma_k} 
\Big) \,
\underbrace{\int \cD [ \Phi_{\tilde \gamma_{[0,n]}} ] \,  \varphi_{\vec \gamma_0} \, \varphi_{(f(\vec \gamma_{n-1}),\gamma_n)}
\Big(
\prod_{k=1}^{n-1}  E_{\gamma_k} 
\Big) \,
e^{\cS_{0} (\Phi_{\tilde \gamma_{[0,n]}})}}_{= \; \pm 1} 
\\
\; & \times 
\underbrace{\int \cD [ \Phi_{\Omega \setminus \tilde \gamma_{[0,n]}} ] \, \varphi_{\vec \gamma_n} \, \varphi_{\halfedge}
\Big(
\prod_{\edual \in \hat P(n)} x_{\edual} E_{\edual} 
\Big) \,
e^{\cS_{0} (\Phi_{\Omega \setminus \tilde \gamma_{[0,n]}})}}_{= \; W_{\Omega\setminus \tilde \gamma_{[0,n]},\ul x}^{\vec \gamma_n,f}(\hat P(n))}
\label{eq:What_final}
\end{split}
\end{align}
where the sign ``$\pm$'' is the same as the one in \cref{eq:What}.
Note that the first integral in~\eqref{eq:What_final} only depends on $\gamma_{[0,n]}$.
Hence, we obtain the asserted identity~\eqref{eq:What_splitting}, 
with a different sign compared to \cref{eq:What} which, however, still depends only on $\gamma_{[0,n]}$.  
\end{proof}

This makes it possible to replace the sum in \cref{eq:F_def_again} with a sum over  $\cP^{\vec\gamma_n,\halfedge}_{\Omega \setminus \tilde \gamma_{[0,n]}}$, and \cref{lem:cP_sum_Grassman_general} lets us rewrite this sum in terms of a Grassmann integral as
\begin{align}
\label{eq:F_Grassman}
F_W(\gamma_{[0,n]};\halfedge)
=  \pm \Big( \prod_{k=1}^{n} x_{\gamma_k} \Big) \;
\int \cD [ \Phi_{\Omega \setminus \tilde \gamma_{[0,n]}} ] \, \varphi_{\vec \gamma_n} \, \varphi_{\halfedge}
\Big(
\prod_{\edual \in \hat P(n)} x_{\edual} E_{\edual} 
\Big) \,
e^{\cS_{\ul x} (\Phi_{\Omega \setminus \tilde \gamma_{[0,n]}})}
.
\end{align} 
In conclusion, we have shown that the martingale observable can be expressed in terms of the Berezin measure~\eqref{eq:correlation_function}:

\begin{prop}
\label{prop:M_free_expectation}
For each admissible boundary condition $\xi = (\bhalfedge,\ehalfedge)$ and half-edge $\halfedge \in \cVcluster\Omega \setminus \{\bhalfedge\}$, 
for the Peierls interface $\gamma$ \textnormal{(}Def.~\ref{def:interface}, with $\tilde \gamma_{[0,n]}$ as in \cref{lem:Cgamma}\textnormal{)} we have
\begin{align}
\label{eq:W_Grassman}
M_W( \halfedge ) [\gamma_{[0,n]}]
=  
\frac{
\big\langle \varphi_{\vec \gamma_n} \, \varphi_{\halfedge} \big\rangle_{\Omega \setminus \tilde \gamma_{[0,n]},\ul x}^{+}
}{
\big\langle \varphi_{\vec \gamma_n} \, \varphi_{\ehalfedge} \big\rangle_{\Omega \setminus \tilde \gamma_{[0,n]},\ul x}^{+}
} .
\end{align} 
\end{prop}

To proceed with the proof of convergence of the interface, 
at this point one would need to discuss the convergence of this ratio in the scaling limit, which follows by relating it to an object with suitable discrete holomorphicity properties.  
It is presumably the case that $M_W$ is sufficiently similar to that of~\cite{Smirnov:Towards_conformal_invariance_of_2D_lattice_models,Chelkak-Izyurov:Holomorphic_spinor_observables_in_critical_Ising_model,Izyurov:Smirnovs_observable_for_free_boundary_conditions_interfaces_and_crossing_probabilities}
to repeat the argument used there, but this involves carefully studying the sign of $W$ (cf.~\cref{rem:W_sign}).

\begin{remark} \label{rem:sholo}
As a simpler alternative for the convergence, let us note that the expressions appearing in \cref{eq:W_Grassman}
are almost the same Berezin expectation values 
as those studied in~\cite{CCK:Revisiting_the_combinatorics_of_the_2D_Ising_model}, 
where they are also expressed in terms of s-holomorphic objects with appropriate boundary conditions.
Slightly more precisely, introducing an antisymmetric matrix $S^{\Omega,\ul x}$ indexed by half-edges 
such that $\cS_{\ul x}(\Phi_\Omega)= -\tfrac12 \big( \Phi_\Omega, S^{\Omega, \ul x} \Phi_\Omega \big)$, 
\begin{align}
	\label{eq:pair_as_inverse}
	\sum_{h_2}
	S^{\Omega, \ul x}_{h_1,h_2}
	\big\langle \varphi_{\vec \gamma_n} \, \varphi_{h_3} \big\rangle_{\Omega \setminus \tilde \gamma_{[0,n]},\ul x}^{+}
	=
	\delta_{h_1 h_3}
	,
\end{align}
and defining the observable 
\begin{align}
	\bfM (e)
	:=
	\begin{cases}
		-\ii M_W( (f,e))
		+ M_W ( (f+1,e))
		, &
		e = \{f,f+1\}
		,
		\\
		-e^{\pi \ii / 4} M_W( (f,e))
		+ e^{- \pi \ii / 4} M_W ( (f+\ii,e))
		, &
		e = \{f,f+\ii\}
		,
	\end{cases}
	\label{eq:complex_M}
\end{align}
one can check, after some tedious algebra using \cref{eq:pair_as_inverse}, that for the \emph{critical edge-weight} $x_e \equiv x_c := \sqrt 2 -1$
this defines a function which is s-holomorphic (in the precise sense of~\cite{Smirnov:Conformal_invariance_in_random_cluster_models1})
for all $e \in \cE_\Omega$ which are not incident to $f(\vec \gamma_n)$.
\end{remark}

\bigskip{}
\section{Martingale observable for non-integrable Ising models}
\label{sec:interact}

We now turn to the version of the Ising model with formal Hamiltonian~\eqref{eq:formal_inter_Ham}, 
which we define more precisely in \cref{eq:Ham_lam_gen} and \cref{lem:relate_Ham_lam_gen}.
We use the following quantities.
\begin{itemize}[leftmargin=*]
\item 
$\Interaction \colon \PowerSet{\cEdual{}} \to \bR$ 
is a finite-range, translation-invariant interaction: 
there exists $R \in (0,\infty)$ for which $\smash{\underset{X \in \supp U}{\max} \diam X = R}$, 
and $\Interaction(X+a) = \Interaction(X)$ for all $a \in \bZ + \ii \bZ$ and $X \subset \cEdual{}$. 

\item 
Fix a finite vertex set $\Omega \subset \cVprimal{}$ 
and an admissible $\xi = (\halfedge_1, \halfedge_2, \ldots, \halfedge_{2N})$.

\item  
Fix $\Ombc \supset \Omega$ and $\Pbc \subset \cEdual{\Ombc\setminus \Omega}$ 
with $|\Ombc| < \infty$ such that $P \cup \Pbc \in \cP_{\Ombc}$ for all $P \in \cP_\Omega^\xi$.  
\end{itemize}

Given a spin configuration $\sigma \in \{\pm1\}^{\cVprimal\Omega}$, define
$\epsilon_{\eprimal}(\sigma) := \tfrac12 (1 - \sigma_{\vprimal} \sigma_{\wprimal})$ for edges $\eprimal = \{ \vprimal,\wprimal \}$ as before, with spins for vertices outside of $\Omega$ fixed according to $\Pbc$ by 
\begin{align*}
\sigma_{\vprimal} 
= (-1)^{|K_{\vprimal} \cap P|} , \qquad \vprimal \in \Ombc \setminus \Omega ,
\end{align*}
where $K_{\vprimal}$ is the set of dual edges crossed by a simple path from $\vprimal$ to $\bZ^2 \setminus \Ombc$ without using 
any edges in $\cEprimal\Omega$ (parity constraints imply its independence of the choice of path).

To define the general (non-integrable) model, we consider the Hamiltonian with boundary condition $\Pbc \in \cP^{\xi}_{\Ombc}$ defined by
\begin{align}
H_{\Omega,J,\lambda}^{\Pbc}(\sigma)
:=
2 J \, 
\sum_{\eprimal \in \cEprimal\Omega} \epsilon_{\eprimal}(\sigma)
+ 
\lambda 
\sum_{X \subset \cEdual\Omega}
\Interaction(X) \prod_{\edual \in X} \epsilon_{\eprimal}(\sigma) ,
\label{eq:Ham_lam_gen}
\end{align}
where $\lambda \in \bR$ is a parameter controlling the strength of the added non-planar or multi-spin interaction 
(thought of as fixed but small in absolute value compared to $J$). 

The Hamiltonian~\eqref{eq:Ham_lam_gen} is equivalent to that of \cref{eq:formal_inter_Ham}, which is the form used in~\cite{GGM:The_scaling_limit_of_the_energy_correlations_in_non_integrable_Ising_models, AGG:Non_integrable_Ising_models_in_cylindrical_geometry, AGG:Energy_correlations_of_non_integrable_Ising_models}, 
since any even local function of the spins can be written in either form up to a constant term.  The relationship can be made explicit as follows.

\begin{lem} \label{lem:relate_Ham_lam_gen}
Suppose that
\begin{align*}
\Interaction(X) = (-2)^{|X|} \sum_{\substack{\Yprimal \subset \cEprimal{\Omega} \\ \Yprimal \supset \Xprimal}} V(\Yprimal) , \qquad X \subset \cEdual{\Omega} ,
\end{align*}
where $X \leftrightarrow \Xprimal$ is the bijection between dual edges in $X \subset \cEdual{\Omega}$ 
and edges in $\Xprimal \subset \cEprimal{\Omega}$.
The Gibbs measures associated to the two Hamiltonians~\eqref{eq:formal_inter_Ham} and~\eqref{eq:Ham_lam_gen} 
agree in the sense that
\begin{align*}
\frac{e^{-\beta \Hprimal_{\Omega, J,\lambda}(\sigma)}}{\underset{\sigma}{\sum} \exp(-\beta \Hprimal_{\Omega, J,\lambda}(\sigma))}
= \frac{e^{-\beta H_{\Omega, J,\lambda}(\sigma)}}{\underset{\sigma}{\sum} \exp(-\beta H_{\Omega, J,\lambda}(\sigma))} ,
\end{align*}
where
\begin{align*}
\Hprimal_{\Omega, J,\lambda}(\sigma)
:= \; & - J \sum_{\{\vprimal,\wprimal\} \in \cEprimal{}} \sigma_{\vprimal} \sigma_{\wprimal}
+
\lambda \sum_{\Yprimal \subset \cEprimal{}} \, V(\Yprimal) 
\prod_{\{\vprimal,\wprimal\} \in \Yprimal} \sigma_{\vprimal} \sigma_{\wprimal} , \\
H_{\Omega, J,\lambda}(\sigma)
:= \; & 2 J \, 
\sum_{\eprimal \in \cEprimal{\Omega}} \epsilon_{\eprimal}(\sigma)
+ 
\lambda 
\sum_{X \subset \cEdual{\Omega}}
\Interaction(X) \prod_{\edual \in X} \epsilon_{\eprimal}(\sigma) .
\end{align*}
\end{lem} 

\begin{proof}
Note that the first terms in the two Hamiltonians can be related as
\begin{align*}
2 J \, \sum_{\eprimal \in \cEprimal{\Omega}} \epsilon_{\eprimal}(\sigma)
= J \sum_{\{\vprimal,\wprimal\} \in \cEprimal{\Omega}} (1 - \sigma_{\vprimal} \sigma_{\wprimal})
= J \, |\cEprimal{\Omega}| - J \sum_{\{\vprimal,\wprimal\} \in \cEprimal{\Omega}} \sigma_{\vprimal} \sigma_{\wprimal} .
\end{align*}
The constant factor $J \, |\cEprimal{\Omega}|$ does not affect the Gibbs measure.
The second (potentially integrability-breaking) term in the two Hamiltonians can be related as follows:
\begin{align*}
\sum_{X \subset \cEdual{\Omega}}
\Interaction(X) \prod_{\edual \in X} \epsilon_{\eprimal}(\sigma)
= \; & \sum_{X \subset \cEdual{\Omega}} (-2)^{|X|} 
\sum_{\substack{\Yprimal \subset \cEprimal{\Omega} \\ \Yprimal \supset \Xprimal}} V(\Yprimal) 
\; 2^{-|X|} \sum_{A \subset X} (-1)^{|A|} \prod_{\{\vprimal,\wprimal\} \in \Aprimal} \sigma_{\vprimal} \sigma_{\wprimal} \\
= \; & \sum_{\substack{A, X, Y \subset \cEdual{\Omega} \\ A \subset X \\ \Xprimal \subset \Yprimal}} 
(-1)^{|A|} \, (-1)^{|X|} \, V(\Yprimal) 
\prod_{\{\vprimal,\wprimal\} \in \Aprimal} \sigma_{\vprimal} \sigma_{\wprimal} .
\end{align*}
For a fixed set $A \subset \cEdual{\Omega}$, the coefficient of $\underset{\{\vprimal,\wprimal\} \in \Aprimal}{\prod} \sigma_{\vprimal} \sigma_{\wprimal}$ is
\begin{align*}
\sum_{\substack{X, Y \subset \cEdual{\Omega} \\ A \subset X \\ \Xprimal \subset \Yprimal}} 
(-1)^{|A|} \, (-1)^{|X|} \, V(\Yprimal) 
= \; & V(\Aprimal) 
+ \sum_{\substack{Y \subset \cEdual{\Omega} \\ A \subsetneq Y}} 
(-1)^{|A|} \, \sum_{\substack{X \subset \cEdual{\Omega} \\ A \subset X \\ \Xprimal \subset \Yprimal}} (-1)^{|X|} 
\; = \; V(\Aprimal) ,
\end{align*}
because the binomial formula shows that
\begin{align*}
\sum_{\substack{X \subset \cEdual{\Omega} \\ A \subset X \\ \Xprimal \subset \Yprimal}} (-1)^{|X|}
= \; & \sum_{k=0}^{|Y \setminus A|} (-1)^{|A|+k} \binom{|Y \setminus A|}{k}
= 0 .
\end{align*}
Hence, we conclude that 
\begin{align*}
\sum_{X \subset \cEdual{\Omega}}
\Interaction(X) \prod_{\edual \in X} \epsilon_{\eprimal}(\sigma)
= \; & \sum_{A \subset \cEdual{\Omega}}  V(\Aprimal) \prod_{\{\vprimal,\wprimal\} \in \Aprimal} \sigma_{\vprimal} \sigma_{\wprimal} ,
\end{align*}
which finishes the proof.
\end{proof}

The Hamiltonian~\eqref{eq:Ham_lam_gen}, like the planar one, 
can be written as a function of the low-temperature contour configuration 
$P(\sigma):= \{ e \in \cEdual\Omega \;|\; \ \epsilon_{\eprimal} (\sigma) = 1\}$, 
which has distribution
\begin{align}
\label{eq:contour_probability_measure_lam}
\bP_{\Omega,x,\beta,\lambda}^{\Pbc}[P] 
:= \frac{1}{Z_{\Omega,x,\beta,\lambda}^{\Pbc}} \, 
x^{|P|} \, \prod_{X \subset P}
e^{- \beta \lambda \Interaction^{\Pbc}(X) }
, \qquad P \in \cP^{\xi}_\Omega ,
\end{align}
corresponding to the Gibbs measure of \cref{eq:Ham_lam_gen}, with 
\begin{align}
\Interaction^{\Pbc}(X) := \sum_{Y \subset \Pbc} \Interaction(X \cup Y) 
\label{eq:UPbc}
\end{align}
and partition function
\begin{align*}
Z_{\Omega,x,\beta,\lambda}^{\Pbc}
= \; &
\sum_{P \in \cP^{\xi}_\Omega} 
x^{|P|} \, \prod_{X \subset P}
e^{- \beta \lambda \Interaction^{\Pbc}(X) } ,
\end{align*}
where $x = e^{- 2\beta J}$ is the (constant) edge weight and $\beta \in (0,\infty)$ is the inverse temperature.
We will again focus on the case where $N=1$ (Dobrushin boundary conditions, recall \cref{fig:P_example}), 
writing $\xi = (\bhalfedge,\ehalfedge)$ although 
the generalization to $N \geq 2$ is straightforward. 

First, the function $W_{\Omega,x}^{\xi}$ defined in~\eqref{eq:W_nonint}
can be generalized\footnote{We identify $x$ with 
the constant function in this and other quantities which depend on the edge weights.} to the $\lambda \neq 0$ case as
\begin{align}
\begin{split}
W_{\Omega,x,\beta,\lambda}^{\xi,\Pbc}(P)
:= \; &
W_{\Omega,x}^{\xi}(P)
\; \prod_{X \subset P}
e^{- \beta \lambda \Interaction^{\Pbc}(X) }
\\
= \; &
W_{\Omega,x}^{\xi}(P)
\; \prod_{X \subset P} 
\Big( 
\prod_{Y \subset \Pbc} 
e^{- \beta \lambda \Interaction(X \cup Y)}
\Big)
, 
\qquad P \subset \cEdual\Omega . 
\label{eq:W_l_def}
\end{split}
\end{align}
This can then be used to define martingale observables as in \cref{sec:interface},
in particular since the second factor does not change the sign.
More precisely, for the Peierls interface started at $\vec \gamma_0 = \bhalfedge$, 
the function defined in \cref{eq:F_def},
\begin{align}
\label{eq:F_def_again_interact}
F_W(\gamma_{[0,n]};\halfedge)
:=
\sum_{P \in \cP_\Omega^{\vec \gamma_0,\halfedge}} 
\one_{C_{\gamma,n}}(P) 
\, W_{\Omega,x,\beta,\lambda}^{\vec \gamma_0,\halfedge,\Pbc}(P) , \qquad h \in \cVcluster\Omega ,
\end{align}
where $C_{\gamma,n} := \{ P \subset \cEdual\Omega \;|\; \gamma_{[0,n]} = \gamma_{[0,n]}(P) \}$, 
gives rise to the martingale observable 
\begin{align}
\label{eq:M_interacting}
M_W( \halfedge ) [\gamma_{[0,n]}]
=
\frac{F_W(\gamma_{[0,n]};\halfedge)}{F_W(\gamma_{[0,n]};\ehalfedge)} , \qquad n \leq \etime .
\end{align}

\subsection{Grassmann integral representation}

The observable~\eqref{eq:M_interacting} also admits a representation as a Grassmann integral, constructed as follows. 
For $P \subset \cEdual\Omega$, we denote the set of collections of non-overlapping subsets of $P$ by 
\begin{align*}
\cH(P) := \big\{ \cQ \subset \PowerSet{P} \; \big| \; \ Q \cap Q' = \emptyset \ \textnormal{ for all } Q,Q' \in \cQ \big\} .
\end{align*}
A collection $\cY \subset \PowerSet{P}$ is said to be \emph{overlap-connected} (o.c.) if 
the graph formed from $\cY$ by adding an edge for each pair $Y_1,Y_2 \in \cY$ with $Y_1 \cap Y_2 \neq \emptyset$ is connected.

\begin{prop}
\label{prop:interacting_Grassmann_main}
For any $\bhalfedge,\halfedge \in \cVcluster\Omega$ and any $\Pbc \subset \cEdual{\Ombc \setminus \cEdual\Omega}$, we have
\begin{align}
\label{eq:interacting_Grassmann_main}
\sum_{P \in \cP_\Omega^{\bhalfedge,\halfedge}}W_{\Omega,x,\beta,\lambda}^{\bhalfedge,\halfedge,\Pbc}(P) 
= \; & 
\int \cD [ \Phi_\Omega ] \, \varphi_{\bhalfedge} \, \varphi_{\halfedge}  \, e^{\cS_{x} (\Phi_\Omega) \, + \, \Potential_{\beta,x, \lambda}^{\Pbc}(\Phi_\Omega)} .
\end{align}
where $\Potential$ is a Grassmann polynomial
of the form
\begin{align*}
\Potential_{\beta,x, \lambda}^{\Pbc}(\Phi_\Omega)
= \; & 
\sum_{X \subset \cEdual\Omega}
\lis{U}_{\beta, x, \lambda}^{\Pbc}(X) 
\Big( \prod_{\edual \in X} E_{\edual} \Big) , \\
\lis{U}_{\beta, x, \lambda}^{\Pbc}(X) 
= \; & x^{|X|} \, \sum_{\substack{\cY \subset \PowerSet{X} \\ \cY \textup{ o.c.} \\ \underset{Y \in \cY}{\bigcup} Y \, = \, X}}
\prod_{Y \in \cY}
\Big( 
\exp( - \beta \lambda \Interaction^{\Pbc}(Y)   )
-1
\Big)
\end{align*}
Furthermore, whenever $X \subset \cEdual\Omega \cap \cEdual{\Omega'}$ and $\Pbc$ and $\Pbc'$ are valid boundary conditions for $\Omega$ and $\Omega'$ respectively, we have
\begin{align}
\lis{U}_{\beta, x, \lambda}^{\Pbc}(X) 
=
\lis{U}_{\beta, x, \lambda}^{\Pbc'}(X) 
\quad \textnormal{ whenever } \quad
\dist(X,\Pbc \triangle \Pbc' )
> R
,
\label{eq:Ubar_boundary}
\end{align}
and there are $U$-dependent constants $C_U , c_U \in (0,\infty)$ such that for $|\beta \lambda| \leq C_U$, 
\begin{align}
\big| \lis{U}_{\beta, x, \lambda}^{\Pbc}(X) \big|
\leq 
C_U \, e^{- c_U \lambda  T(X)} 
,
\label{eq:Ubar_decay}
\end{align}
where $T(X)$ denotes the minimum of $|Z|$ over connected sets $Z \subset \cEdual{}$ such that $X \subset Z$. 
\end{prop}

Note that we do not assume any particular relationship between $\Pbc$ and $\bhalfedge, \halfedge$.

\begin{proof}
Fix $P \subset \cEdual\Omega$.
Expanding as a binomial, we can write 
\begin{align}
\nonumber
\prod_{X \subset P}
\exp( - \beta \lambda \Interaction^{\Pbc}(X)   )
= \; & 
\sum_{\cX \subset \PowerSet{P}} 
\prod_{X \in \cX}
\Big( 
\exp( - \beta \lambda \Interaction^{\Pbc}(X)   )
-1
\Big)
\\ 
\nonumber
= \; &
\sum_{\cQ \in \cH(P)}
\prod_{Q \in \cQ}
\sum_{\substack{\cY \subset \PowerSet{Q} \\ \cY \textup{ o.c.} \\ \underset{Y \in \cY}{\bigcup} Y \, = \, Q}}
\prod_{Y \in \cY}
\Big( 
\exp( - \beta \lambda \Interaction^{\Pbc}(Y)   )
-1
\Big)
\\ 
=: \; & 
\sum_{\cQ \in \cH(P)}
\prod_{Q \in \cQ}
\tilde \Interaction_{\beta, \lambda}^{\Pbc}(Q)
.
\label{eq:polymer}
\end{align}
Using the definition from \cref{eq:W_l_def}, we obtain 
\begin{align}
\label{eq:W_inter_bigsum}
\; & 
\sum_{P \in \cP_\Omega^{\bhalfedge,\halfedge}}
W_{\Omega,x,\beta,\lambda}^{\bhalfedge,\halfedge,\Pbc}(P) 
\\
\nonumber
& = \;  
\sum_{P \in \cP_\Omega^{\bhalfedge,\halfedge}}
W_{\Omega,x}^{\bhalfedge,\halfedge}(P)
\; \prod_{X \subset P}
\exp( - \beta \lambda \Interaction^{\Pbc}(X) )
&&\textnormal{[by~\eqref{eq:W_l_def}]}
\\
\nonumber
& = \;  
\sum_{P \in \cP_\Omega^{\bhalfedge,\halfedge}}
W_{\Omega,x}^{\bhalfedge,\halfedge}(P)
\sum_{\cQ \in \cH(P)}
\prod_{Q \in \cQ}
\tilde \Interaction_{\beta, \lambda}^{\Pbc}(Q)
&&\textnormal{[by~\eqref{eq:polymer}]}
\\
\nonumber
& = \;  
\sum_{P \in \cP_\Omega^{\bhalfedge,\halfedge}}
\Big( \prod_{Q \in \cQ}
\tilde \Interaction_{\beta, \lambda}^{\Pbc}(Q) \Big)
\sum_{\substack{S \in \cP_\Omega^{\bface,\halfedge} \\ S \, \supset \, \underset{Q \in \cQ}{\bigcup} Q}} W_{\Omega,x}^{\bface,\halfedge}(S) 
\\
\nonumber
& = \;  
\int \cD [ \Phi_\Omega ] \, \varphi_{\bhalfedge} \, \varphi_{\halfedge} 
\sum_{\cQ \in \cH(\cEdual\Omega)}
\Big( \prod_{Q \in \cQ}
\tilde \Interaction_{\beta, \lambda}^{\Pbc}(Q) 
\Big(
\prod_{\edual \in  Q} x E_{\edual} 
\Big) \Big) \,
e^{\cS_{x} (\Phi_\Omega)} .
&&\textnormal{[by \cref{lem:cP_sum_Grassman_general}]}
\end{align}
Now, since $\cH(\cEdual\Omega)$ excludes only collections with overlapping sets, we may replace the sum over $\cH(\cEdual\Omega)$ with the sum over all of $\PowerSet{\cEdual\Omega}$,  
since this yields only 
additional terms involving repeated Grassmann variables (which vanish).  
Doing so, we see that
\begin{align*}
\sum_{\cQ \in \cH(\cEdual\Omega)}
\Big( \prod_{Q \in \cQ}
\tilde \Interaction_{\beta, \lambda}^{\Pbc}(Q) 
\Big(
\prod_{\edual \in  Q} x E_{\edual} 
\Big) \Big)
= \; & \sum_{\cQ \in \PowerSet{\cEdual\Omega}}
\Big( \prod_{Q \in \cQ}
\tilde \Interaction_{\beta, \lambda}^{\Pbc}(Q) 
\Big(
\prod_{\edual \in  Q} x E_{\edual} 
\Big) \Big) \\
= \; & 
\prod_{X \subset \cEdual\Omega}
\Big( 1 
+ 
\tilde \Interaction_{\beta, \lambda}^{\Pbc}(X) 
\prod_{\edual \in X} x E_{\edual} 
\Big)
\\
=: \; &
\prod_{X \subset \cEdual\Omega}
\Big( 1 
+ 
\lis{U}_{\beta, x, \lambda}^{\Pbc}(X) 
\prod_{\edual \in X} E_{\edual} 
\Big)
=  \;
\exp \big( \Potential_{\beta, x,\lambda}^{\Pbc}(\Phi_\Omega) \big) ,
\end{align*}
and inserting this in \cref{eq:W_inter_bigsum} gives the asserted identity~\eqref{eq:interacting_Grassmann_main}.

To verify the partial independence property in \cref{eq:Ubar_boundary}, first note from the definition in \cref{eq:UPbc} that $U^{\Pbc}(Y)$ has the same property for all $Y \subset X$, which together with \cref{eq:polymer} implies the same for $\tilde \Interaction_{\beta, \lambda}^{\Pbc}(X)$ and so also for 
\begin{align} \label{eq: Ubar_sum}
\lis{U}_{\beta, x, \lambda}^{\Pbc}(X)
= x^{|X|} \, \tilde \Interaction_{\beta, \lambda}^{\Pbc}(X)
=
\sum_{\substack{\cY \subset \PowerSet{X} \\ \cY \textup{ o.c.} \\ \underset{Y \in \cY}{\bigcup} Y \, = \, X}}
\prod_{Y \in \cY}
x^{|Y|}
\Big( 
\exp( - \beta \lambda \Interaction^{\Pbc}(Y)   )
-1
\Big) .
\end{align}
The bound in \cref{eq:Ubar_boundary} follows by noting that the summand in~\eqref{eq: Ubar_sum} vanishes unless 
\begin{align*}
\diam Y \le R, \textnormal{ for all } Y \in \cY ,
\textnormal{ and thus, also }
R^2 |\cY| \ge T(X) .
\end{align*}
It thus follows that the sum has no more than $R^{2|X|} \le \exp(2 T(X) \log R)$ terms, 
each of which, after noting that $x = e^{-2 \beta J} \le 1$, 
is bounded by $( (e-1) \beta \lambda)^{T(X)/R^2}$ as long as $|\beta \lambda| \le (\underset{Y \subset \cEdual{}}{\max} |\Interaction(Y)|)^{-1}$. 
This gives the final bound~\eqref{eq:Ubar_decay}.
\end{proof}

We can use \cref{prop:interacting_Grassmann_main} to express the function $F_W$ defined 
in \cref{eq:F_def_again_interact} as a Grassmann integral as well.
This leads to an expression for the martingale in \cref{thm:main_thm}.

\begin{lem}
\label{lem:F_interacting_Grassmann}
Whenever $\halfedge \in \cVcluster\Omega \notin \{\vec \gamma_0\}$, we have
\begin{align}
\label{eq:F_interacting_Grassmann}
F_W(\gamma_{[0,n]};\halfedge)
= \; & \pm 
x^n
\int \cD [ \Phi_{\Omega \setminus \tilde \gamma_{[0,n]}} ] \, \varphi_{\vec \gamma_0} \, \varphi_{\halfedge}  \, e^{\cS_{x} (\Phi_{\Omega \setminus \tilde  \gamma_{[0,n]}}) \, + \, \Potential_{\beta,x, \lambda}^{\check \Pbc(n)}(\Phi_{\Omega \setminus \tilde  \gamma_{[0,n]}})}
,
\end{align}
where $\check \Pbc(n) := \Pbc \cup \{ \gamma_1,\ldots,\gamma_n \}$,
with a sign depending only on $\gamma_{[0,n]}$.
\end{lem}
 
\begin{proof}
Recalling the definition of $F_W$, with an appropriate sign, we obtain 
\begin{align*}
\; & F_W  (\gamma_{[0,n]};\halfedge) \\
= \; &
\sum_{P \in \cP_\Omega^{\vec \gamma_0,\halfedge} \cap C_{\gamma,n}}
\, W_{\Omega,x,\beta,\lambda}^{\vec \gamma_0,\halfedge,\Pbc}(P) ,
&&\textnormal{[by~\cref{eq:F_def_again_interact}]}
\\ 
= \; &
\sum_{P \in \cP_\Omega^{\vec \gamma_0,\halfedge} \cap C_{\gamma,n}}
W_{\Omega,x}^{\bhalfedge,\halfedge}(P)
\; \prod_{X \subset P}
\prod_{Y \subset \Pbc} 
e^{ - \beta \lambda \Interaction(X \cup Y) }
&&\textnormal{[by~\cref{eq:W_l_def}]}
\\
= \; & \pm
x^n
\sum_{\hat P \in \cP_{\Omega \setminus \tilde \gamma_{[0,n]}}^{\vec \gamma_n,\halfedge}} 
 W_{\Omega\setminus \tilde \gamma_{[0,n]},x}^{\vec \gamma_n,\halfedge}(\hat P) 
\; \prod_{X \subset \hat P} 
\prod_{Y \subset \check \Pbc(n)} 
e^{ - \beta \lambda \Interaction(X \cup Y) }
&&\textnormal{[by \cref{lem:What_splitting}]}
\\ 
= \; & \pm
x^n
\sum_{\hat P \in \cP_{\Omega \setminus \tilde \gamma_{[0,n]}}^{\vec \gamma_n,\halfedge}} 
W_{\Omega\setminus \tilde \gamma_{[0,n]},x,\beta,\lambda}^{\vec \gamma_n,\halfedge,\check \Pbc(n)}(\hat P) .
&&\textnormal{[by~\cref{lem:Cgamma}]}
\end{align*} 
The asserted identity~\eqref{eq:F_interacting_Grassmann}  
follows by applying \cref{prop:interacting_Grassmann_main} on $\Omega\setminus \tilde \gamma_{[0,n]}$.
\end{proof}

Generalizing the Berezin measure introduced in \cref{eq:correlation_function} to the general case,
\begin{align}
\state{\cO}^{\Pbc}_{\Omega,x,\beta,\lambda}
:=
\frac{
\int \cD [ \Phi_\Omega ] \, \cO  \, e^{\cS_{x} (\Phi_\Omega) \, + \, \Potential_{\beta,x, \lambda}^{\Pbc}(\Phi_\Omega)} .
}{
\int \cD [ \Phi_\Omega ] \,  e^{\cS_{x} (\Phi_\Omega) \, + \, \Potential_{\beta,x, \lambda}^{\Pbc}(\Phi_\Omega)} ,
}
\label{eq:interacting_Berezin_measure}
\end{align}
for $\cO$ in the Grassmann algebra $\Phi_\Omega$, we can rewrite $M_W$ as defined in \cref{eq:M_interacting}: 

\begin{thm}
\label{thm:main_thm}
For any $\halfedge \in \cVcluster\Omega \setminus \{ \vec \gamma_0 \}$, the process
\begin{align*}
n \quad \longmapsto \quad 
\frac{\big\langle \varphi_{\vec \gamma_n} \, \varphi_{\halfedge}  \big\rangle_{\Omega \setminus \tilde \gamma_{[0,n]},x,\beta,\lambda}^{\check \Pbc(n)}}{\big\langle \varphi_{\vec \gamma_n} \, \varphi_{\ehalfedge}  \big\rangle_{\Omega \setminus \tilde \gamma_{[0,n]},x,\beta,\lambda}^{\check \Pbc(n)}} , \qquad n \leq \etime ,
\end{align*}
is a local martingale with respect to the natural filtration generated by $\gamma$.  
\end{thm}
\begin{proof}
\Cref{eq:M_interacting} defines a local martingale by \cref{prop:M_general}.
We have
\begin{align*}
M_W( \halfedge ) [\gamma_{[0,n]}]
= \; &
\frac{F_W(\gamma_{[0,n]};\halfedge)}{F_W(\gamma_{[0,n]};\ehalfedge)} 
= 
\frac{\big\langle \varphi_{\vec \gamma_n} \, \varphi_{\halfedge}  \big\rangle_{\Omega \setminus \tilde \gamma_{[0,n]},x,\beta,\lambda}^{\check \Pbc(n)}}{\big\langle \varphi_{\vec \gamma_n} \, \varphi_{\ehalfedge}  \big\rangle_{\Omega \setminus \tilde \gamma_{[0,n]},x,\beta,\lambda}^{\check \Pbc(n)}} 
\end{align*}
by \cref{lem:F_interacting_Grassmann} and \cref{eq:interacting_Berezin_measure}.
\end{proof}

\subsection{Conjecture about the scaling limit}
\label{sec:conjecture}

Let us now present a precise conjecture concerning the martingale observable, 
which would imply its convergence in a locally uniform fashion to the same scaling limit as in the well-studied planar case.
We consider the non-integrable model with the Hamiltonian~\eqref{eq:Ham_lam_gen}.

\begin{conj}
\label{conj:interacting_convergence}
There exists $\lambda_0 = \lambda_0(\Interaction) > 0$ such that for all $|\lambda| \le \lambda_0$,  
there exists $\beta^* = \beta^*(\lambda)$ such that for all finite $\Omega \subset \cVprimal{}$ 
and for all $\xi = (\halfedge_1, \halfedge_2, \ldots, \halfedge_{2N})$ 
and associated external boundary conditions $\Pbc \subset \cEdual{\Ombc\setminus \Omega}$, 
there exists a collection $\{\zeta_{\Omega,\lambda}^{\Pbc}(\halfedge) \;|\; \halfedge \in \cVcluster\Omega\}$, 
of real numbers such that 
the Berezin measure~\eqref{eq:interacting_Berezin_measure} satisfies 
\begin{align}\label{eq:convergence_conjecture} 
\begin{split}
\big\langle \varphi_{\halfedge_1} \, \varphi_{\halfedge_2} \, \cdots \, \varphi_{\halfedge_{2N}}  \big\rangle_{\Omega,x^*,\beta^*,\lambda}^{\Pbc}
= \; & \zeta_{\Omega,\lambda}^{\Pbc}(\halfedge_1) \, \cdots \, \zeta_{\Omega,\lambda}^{\Pbc}(\halfedge_{2N})
\; \big\langle \varphi_{\halfedge_1} \, \varphi_{\halfedge_2} \, \cdots \, \varphi_{\halfedge_{2N}} \big\rangle_{\Omega,x_c}^{+} \\
\; & + R_{\Omega,x^*,\beta^*,\lambda}^{\Pbc}
(\xi) ,
\end{split}
\end{align}
where $\big\langle \cdot \big\rangle_{\Omega,x_c}^{+}$ is the Berezin measure from \cref{eq:correlation_function}, and
\begin{enumerate}
\item $x^* = e^{-2 \beta^* J}$\textnormal{;}
\item $x_c = \sqrt{2}-1$  is the critical isotropic weight for the Ising model on the square lattice\textnormal{;}
\item $\zeta_{\Omega,\lambda}^{\Pbc}$ are uniformly bounded and there exist constants $\theta > 0$ and $\zeta_{\bulk,\lambda} \in \bR$ such that
\begin{align}
\label{eq:zeta_bulk_convergence}
\zeta_{\Omega,\lambda}^{\Pbc}(\halfedge) = \zeta_{\bulk,\lambda} + O( (\dist(\halfedge,\partial\Omega))^{-\theta})
\end{align}
uniformly on $\Omega$ and $\Pbc$ for sufficiently large $\dist(\halfedge,\partial\Omega)$\textnormal{;}
\item for any two different domains $\Omega$ and $\Omega'$ and associated external boundary conditions $\Pbc$ and $\Pbc'$, we have
\begin{align}
\label{eq:zeta_boundary_dependence}
\zeta_{\Omega,\lambda}^{\Pbc}(\halfedge) - \zeta_{\Omega',\lambda}^{\Pbc'}(\halfedge) 
= O( (\dist(\halfedge,\Omega \triangle \Omega'))^{-\theta}) + O( (\dist(\halfedge,\Pbc \triangle \Pbc'))^{-\theta}) ;
\end{align}
\item and $R_{\Omega,x^*,\beta^*,\lambda}^{\Pbc}
(\xi) = O \big( \underset{1 \leq i \neq j \leq 2N}{\min} |\halfedge_i - \halfedge_j|^{-1-\theta} \big)$ uniformly on $\Omega$ and $\Pbc$.
\end{enumerate}
\end{conj}

\begin{remark}
Importantly, 
\cref{eq:convergence_conjecture} has the same form as existing results about correlation functions of the models under 
consideration~\cite{GGM:The_scaling_limit_of_the_energy_correlations_in_non_integrable_Ising_models, 
AGG:Energy_correlations_of_non_integrable_Ising_models, 
CGG:In_prep}, 
except that here the prefactors are not constant. 
In the special case where all of the half-edges $\halfedge_1,\ldots,\halfedge_{2N}$ lie on the boundary of the half-plane, 
\cref{conj:interacting_convergence} should follow from a straightforward generalization 
of the proof in~\cite{CGG:In_prep}\footnote{The important difference being that~\cite{CGG:In_prep} 
is formulated for open boundary conditions (so using high-temperature rather than low-temperature contours) 
leading to a different \Potential \; --- however, all of the important properties used in the proof are the same.}.
In such a constructive renormalization group treatment, the prefactors $\zeta_{\Omega,\lambda}^{\Pbc}(\halfedge)$ can be given as sums over collections of geometrical objects which contact $\halfedge$, 
which give a contribution decaying in their diameter and depend on the (difference in) boundary conditions only if they cross the (relevant portion of the) boundary, 
yielding \cref{eq:zeta_bulk_convergence,eq:zeta_boundary_dependence}. 
This identification of the terms affected by the boundary conditions is also along similar lines to that 
of~\cite{Greenblatt:Discrete_and_zeta-regularized_determinants_of_the_Laplacian_on_polygonal_domains_with_Dirichlet_boundary_conditions} 
(though the geometrical objects considered there are 
trajectories of random walks or Brownian motions rather than the more varied collections 
of objects which would be involved here).
\end{remark}

In conclusion, applying \cref{conj:interacting_convergence} to the process in \cref{thm:main_thm}, we see that
\begin{align*}
\; &  \frac{\big\langle \varphi_{\vec \gamma_n} \, \varphi_{v_{\halfedge}}  \big\rangle_{\Omega \setminus \tilde \gamma_{[0,n]},x^*,\beta^*,\lambda}^{\check \Pbc(n)}}{\big\langle \varphi_{\vec \gamma_n)} \, \varphi_{v_{\ehalfedge}}  \big\rangle_{\Omega \setminus \tilde \gamma_{[0,n]},x^*,\beta^*,\lambda}^{\check \Pbc(n)}}
\\
= \; &  
\frac{
\zeta_{\Omega \setminus \tilde \gamma_{[0,n]},\lambda}^{\check \Pbc(n)}(\vec \gamma_n) \, 
\zeta_{\Omega \setminus \tilde \gamma_{[0,n]},\lambda}^{\check \Pbc(n)}(\halfedge) \, 
\big\langle \varphi_{\vec \gamma_n} \, \varphi_{\halfedge}  \big\rangle_{\Omega,x_c}^{+}  
+ R_{\Omega  \setminus \tilde \gamma_{[0,n]},x^*,\beta^*,\lambda}^{\check \Pbc(n)}(\vec \gamma_n,\halfedge) 
}{
\zeta_{\Omega \setminus \tilde \gamma_{[0,n]},\lambda}^{\check \Pbc(n)}(\vec \gamma_n) \, 
\zeta_{\Omega \setminus \tilde \gamma_{[0,n]},\lambda}^{\check \Pbc(n)}(\ehalfedge) \, 
\big\langle \varphi_{\vec \gamma_n} \, \varphi_{\ehalfedge}  \big\rangle_{\Omega,x_c}^{+}  
+ R_{\Omega  \setminus \tilde \gamma_{[0,n]},x^*,\beta^*,\lambda}^{\check \Pbc(n)}(\vec \gamma_n,\ehalfedge)
}
\\
= \; & 
\frac{
\zeta_{\Omega \setminus \tilde \gamma_{[0,n]},\lambda}^{\check \Pbc(n)}(\halfedge) 
}{
\zeta_{\Omega \setminus \tilde \gamma_{[0,n]},\lambda}^{\check \Pbc(n)}(\ehalfedge)
}
\frac{\big\langle \varphi_{v(\vec \gamma_n,\gamma_{n})} \, \varphi_{\halfedge}  \big\rangle_{\Omega,x_c}^{+}}{\big\langle \varphi_{v(\vec \gamma_n,\gamma_{n})} \, \varphi_{\ehalfedge}  \big\rangle_{\Omega,x_c}^{+}}
+ O([\dist(\vec \gamma_n,\halfedge)]^{-\theta})
+ O([\dist(\vec \gamma_n,\ehalfedge)]^{-\theta})
\\
= \; & 
\frac{
\zeta_{\bulk,\lambda} \, 
}{
\zeta_{\Omega ,\lambda}^{\Pbc}(\ehalfedge) \, 
}
\frac{\big\langle \varphi_{v(\vec \gamma_n,\gamma_{n})} \, \varphi_{\halfedge}  \big\rangle_{\Omega,x_c}^{+}}{\big\langle \varphi_{v(\vec \gamma_n,\gamma_{n})} \, \varphi_{\ehalfedge}  \big\rangle_{\Omega,x_c}^{+}}
\; + \; O([\dist(\tilde \gamma_{[0,n]} \cup \partial \Omega,\halfedge)]^{-\theta})
\; + \; O([\dist(\tilde \gamma_{[0,n]},\ehalfedge)]^{-\theta})
\end{align*}
in a locally uniform fashion,
so that \cref{conj:interacting_convergence} is sufficient to imply that (up to an innocuous prefactor) 
the martingale observable in \cref{thm:main_thm} has the same scaling limit as the corresponding one for the planar Ising model, 
and similarly for the complex linear combination introduced in \cref{eq:complex_M}.
Along with the straightforward generalization to the case of more marked points, this combines with~\cite[Lemma~3.2]{Izyurov:Critical_Ising_interfaces_in_multiply_connected_domains} to show that \cref{conj:interacting_convergence} implies that the same convergence result holds, with the same limit.  Since all of the other technical lemmas in~\cite{Izyurov:Critical_Ising_interfaces_in_multiply_connected_domains} refer to the continuum spinor which characterizes this limit, the rest of the proof of~\cite[Theorem~1.1]{Izyurov:Critical_Ising_interfaces_in_multiply_connected_domains} can be taken over 
to conclude that this approach is sufficient to prove convergence of the driving function.

Moreover, by construction the discrete interface process is reversible.
Hence, as pointed out by Sheffield~\&~Sun in~\cite[Theorem~1.1]{Sheffield-Sun:Strong_path_convergence_from_Loewner_driving_function_convergence}, 
with some additional technical work this will also imply convergence of the interface in a strong sense in a space of curves.

\bigskip{}
\section{Concluding remarks}
\label{sec:conclusions}

The program of proving \cref{conj:interacting_convergence} in full generality suffers from at least one significant technical obstacle, 
concerning extending the existing rigorous renormalization group treatment 
from the case of a straight boundary considered 
in~\cite{Cava:PhD_thesis, AGG:Energy_correlations_of_non_integrable_Ising_models, CGG:In_prep} 
to very irregular domains associated with the interface process. 
This is due to the decomposition of the ``propagator''  $g_\Omega(h,h') :=\left\langle \varphi_h \varphi_{h'} \right\rangle^+_{\Omega,x_c}$ into components describing 
the behavior on different scales in a way compatible with Gram decomposition, as in for example~\cite[Proposition~2.3]{AGG:Non_integrable_Ising_models_in_cylindrical_geometry}.
So far, using methods based on variants of Fourier analysis (namely, implementing what is known in the physics literature as a ``momentum shell'' cutoff), 
the decomposition has been established --- but only on domains with a high degree of symmetry (such as a discrete torus or cylinder).

There are at least two possible strategies to proceed to overcome this technical issue. 
Using lattice decimation, a suitable Gram decomposition is available almost trivially~\cite[Section~2]{Dimock-Yuan:Structural_stability_of_the_RG_flow_in_GN_model},
so it should be possible to proceed in this way as long as the other estimates necessary can be obtained.
Alternatively, one may proceed by proving \cref{conj:interacting_convergence} only on domains with a high degree of symmetry,
which should be readily possible using the RG technology. Thereafter, provided that the convergence is sharp enough, 
one might be able to invoke locality arguments to lift the conclusion to rough domains~\cite{Mahfouf:Private}.

Note also that our strategy depends a great deal on the comparison to the well-understood nearest-neighbor case, 
and thus also on its discrete holomorphicity properties.  
More generally, conformal invariance in two dimensions is expected to be an emergent property of critical scaling limits~\cite{Zamolodchikov:RG_and_perturbation_theory_about_fixed_points_in_2D_field_theory,
Polchinski:Scale_and_conformal_invariance_in_quantum_field_theory,
Nakayama:Scale_invariance_vs_conformal_invariance} 
even in cases where there is no such explicit, exactly solved model to fall back on in this way.
A prominent example of this is the interacting dimer model, whose scaling limit is qualitatively different from the (exactly solved) non-interacting case~\cite{GMT:Haldane_relation_for_interacting_dimers}. 
Here, one still expects it to be possible to prove conformal invariance as a consequence of scale-invariance via renormalization group techniques~\cite{Nakayama:Scale_invariance_vs_conformal_invariance}.  
This however clearly requires a continuous RG transform -- so lattice decimation is unsatisfactory, and it will be necessary to find an alternative decomposition, 
e.g., based on expressing the propagator in terms of the heat kernel or a resolvent.

\bigskip{}
\bibliographystyle{alpha}

\newcommand{\etalchar}[1]{$^{#1}$}

\end{document}